\documentclass{llncs}
\usepackage{fullpage}
\usepackage{amssymb}
\usepackage{amsmath}
\usepackage[dvipsnames]{xcolor}
\usepackage{breakcites}
\usepackage{mdframed}
\usepackage{algorithm}
\usepackage{algpseudocode}
\usepackage[colorlinks=true,citecolor=NavyBlue,linkcolor=Black,backref]{hyperref}
\usepackage{nicefrac}
\usepackage{csquotes}
\usepackage{mathtools}
\usepackage{xspace}
\usepackage{booktabs}
\usepackage{mleftright}
\usepackage{orcidlink}
\usepackage[adversary,sets,asymptotics,landau,lambda,operators,logic,primitives]{cryptocode}
\allowdisplaybreaks
\makeatletter
\spnewtheorem{claim}[theorem]{Claim}{\bfseries}{\itshape}
\let\c@lemma\relax
\spnewtheorem{lemma}[theorem]{Lemma}{\bfseries}{\itshape}
\let\c@corollary\relax
\spnewtheorem{corollary}[theorem]{Corollary}{\bfseries}{\itshape}
\let\c@proposition\relax
\spnewtheorem{proposition}[theorem]{Proposition}{\bfseries}{\itshape}
\spnewtheorem{assumption}[definition]{Assumption}{\bfseries}{\itshape}
\makeatother
\usepackage{tikz}
\renewcommand{\vec}[1]{\boldsymbol{#1}}

\usepackage{pifont}

\mathchardef\mhyphen="2D

\newcommand{\E}{\mathbb{E}}
\newcommand{\cG}{\mathcal{G}}
\newcommand{\cL}{\mathcal{L}}
\newcommand{\cU}{\mathcal{U}}

\newcommand{\init}{\mathsf{Init}}
\newcommand{\peel}{\mathsf{Delete}}
\newcommand{\encode}{\mathsf{Insert}}
\newcommand{\decode}{\mathsf{ListEntries}}
\newcommand{\basicinit}{\mathsf{BasicInit}}
\newcommand{\basicpeel}{\mathsf{BasicDelete}}
\newcommand{\basicencode}{\mathsf{BasicInsert}}
\newcommand{\basicdecode}{\mathsf{BasicListEntries}}

\newcommand{\universe}{\mathcal{U}}
\newcommand{\keyspace}{\mathcal{K}}
\DeclareMathOperator\lglg{lglg}

\newcommandx*{\pcand}[2][1=\ ,2=\ ]{#1\highlightkeyword[#2]{and}}
\newcommand{\ssdiff}{\ensuremath{\bigtriangleup}\xspace}

\makeatletter

\newenvironment{customTheorem}[1]
  {\count@\c@theorem
   \global\c@theorem#1 %
    \global\advance\c@theorem\m@ne
   \theorem}
  {\endtheorem
   \global\c@theorem\count@}

\makeatother

\DeclareMathOperator{\polylog}{\mathsf{polylog}}

\newcommand{\sk}{\mathsf{sk}}
\newcommand{\pk}{\mathsf{pk}}

\newcommand{\encscheme}{\mathcal{E}}
\newcommand{\gen}{\mathsf{Gen}}
\renewcommand{\sample}{\mathsf{Sample}}
\newcommand{\entry}{\mathsf{Entry}}
\newcommand{\recindex}{\mathsf{Index}}
\newcommand{\eval}{\mathsf{Eval}}
\newcommand{\compress}{\mathsf{Compress}}
\newcommand{\decompress}{\mathsf{Decompress}}

\newcommand{\msgspace}{\mathcal{M}}
\newcommand{\circspace}{\mathcal{Z}}
\newcommand{\ipcirc}{\circspace_{\mathsf{ip}}}

\newcommand{\pkspace}{\mathcal{P}}
\newcommand{\cspace}{\mathcal{C}}
\newcommand{\filtered}{\mathcal{F}}

\DeclareMathOperator{\hw}{\mathsf{hw}}
\newcommand{\valid}{\mathsf{vld}}
\newcommand{\sparse}{\mathsf{sparse}}
\newcommand{\bintofield}{\mathsf{binToField}}
\newcommand{\fieldtobin}{\mathsf{fieldToBin}}

\newcommand{\proj}{\mathsf{proj}}
\newcommand{\decomp}{\mathsf{decomp}}
\newcommand{\unique}{\raisebox{-.25ex}{\textrm{\ding{36}}}}

\usepackage{scalerel}

\newcommand{\basicdelete}{\mathsf{BasicDel}}
\renewcommand{\gamechange}[2][gamechangecolor]{%
{\setlength{\fboxsep}{0pt}\ifmmode%
\mathchoice{%
\colorbox{#1}{$\displaystyle#2$}%
}{%
\colorbox{#1}{$#2$}%
}{%
\colorbox{#1}{$\scriptstyle#2$}%
}{%
\colorbox{#1}{$\scriptscriptstyle#2$}%
}%
\else{#2}\fi}%
}

\title{Invertible Bloom Lookup Tables\texorpdfstring{\\}{} with Less Memory and Randomness}

\author{
	Nils Fleischhacker \inst{1}\thanks{\texttt{mail@nilsfleischhacker.de}. Funded by the Deutsche Forschungsgemeinschaft (DFG, German Research Foundation) under Germany's Excellence Strategy - EXC 2092 CASA - 390781972.}\orcidlink{0000-0002-2770-5444}\and
        Kasper Green Larsen\inst{2}\thanks{\texttt{larsen@cs.au.dk}. Supported by a DFF Sapere Aude Research Leader grant No 9064-00068B.}\orcidlink{0000-0001-8841-5929} \and 
	Maciej Obremski\inst{3}\thanks{\texttt{obremski.math@gmail.com}. Funded by MOE2019-T2-1-145 Foundations of quantum-safe cryptography.}\orcidlink{0000-0003-4174-0438} \and
        Mark Simkin\inst{4}\thanks{\texttt{mark@univariate.org}}\orcidlink{0000-0002-7325-5261}
}
\institute{
  Ruhr University Bochum \and
  Aarhus University \and
  National University of Singapore \and
  Independent Researcher
}

\begin{document}
\pagestyle{plain}

\maketitle

\begin{abstract}
In this work we study Invertible Bloom Lookup Tables (IBLTs) with small failure probabilities.
IBLTs are highly versatile data structures that have found applications in set reconciliation protocols, error-correcting codes, and even the design of advanced cryptographic primitives.
For storing $n$ elements and ensuring correctness with probability at least $1 - \delta$, existing IBLT constructions require $\Omega(n(\frac{\log(1/\delta)}{\log(n)}+1))$ space and they crucially rely on fully random hash functions.

We present new constructions of IBLTs that are simultaneously more space efficient and require less randomness.
For storing $n$ elements with a failure probability of at most $\delta$, our data structure only requires $\mathcal{O}\left(n + \log(1/\delta)\log\log(1/\delta)\right)$ space and $\mathcal{O}\left(\log(\log(n)/\delta)\right)$-wise independent hash functions.

As a key technical ingredient we show that hashing $n$ keys with any $k$-wise independent hash function $h:U \to [Cn]$ for some sufficiently large constant $C$ guarantees with probability $1 - 2^{-\Omega(k)}$ that at least $n/2$ keys will have a unique hash value.
Proving this is non-trivial as $k$ approaches $n$.
We believe that the techniques used to prove this statement may be of independent interest.

We apply our new IBLTs to the encrypted compression problem, recently studied by Fleischhacker, Larsen, Simkin (Eurocrypt 2023).
We extend their approach to work for a more general class of encryption schemes and using our new IBLT we achieve an asymptotically better compression rate.

\end{abstract}

\thispagestyle{empty}
\newpage
\setcounter{page}{1}


\section{Introduction}\label{sec:intro}

The Invertible Bloom Lookup Table (IBLT) is a very elegant data structure by Goodrich and Mitzenmacher~\cite{All:GooMit11}. It functions much like a dictionary data structure, supporting insertions, deletions and the retrieval of key-value pairs. What is special about the IBLT, is that upon initialization, one decides on a threshold $n$. Now, regardless of how many key-value pairs are present in the IBLT, the space usage will always remain proportional to $n$. Of course this comes at a cost, namely that the retrieval operations will temporarily stop functioning, when the number of pairs stored in the IBLT exceeds $n$. When the number of stored pairs falls below $n$ again, the IBLT will resume supporting retrieval queries.

The above functionality is extremely useful in many applications. Consider for instance the set reconciliation problem~\cite{IEEE:MTZ03,ACM:EGUV11}. Here two parties Alice and Bob hold sets $S_A$ and $S_B$ of key-value pairs. Think of these sets as two replicas of a database storing key-value pairs. In applications where insertions and deletions into the database must be supported quickly, we may allow the two sets $S_A$ and $S_B$ to be slightly inconsistent, such that a client performing an operation on the database will not have to wait for synchronization among the two replicas. Instead, Alice and Bob will every now and then synchronize their two sets $S_A$ and $S_B$. For this purpose, Alice maintains an IBLT for her set $S_A$, which she may send to Bob. Upon receiving the IBLT, Bob then deletes every element from his set $S_B$ from Alice's IBLT. If $|(S_A \setminus S_B) \cup (S_B \setminus S_A)|$ is less than the threshold $n$, Bob can retrieve the key-value pairs in $S_A \setminus S_B$. Since the space usage of IBLTs is only proportional to the threshold $n$, this allows for the communication between Alice and Bob to be proportional to $|(S_A \setminus S_B) \cup (S_B \setminus S_A)|$ and not $|S_A|$ or $|S_B|$. This may result in significant savings, when the sets $S_A$ and $S_B$ are large, but very similar.
IBLTs have also seen uses in numerous other applications, ranging from distributed systems applications~\cite{PABHL17,DC:MP18} over fast error-correcting codes~\cite{IEEE:MV12} to cryptography~\cite{ACNS:AGLPPT17,EC:FleLarSim22,EC:FleLarSim23}.

The surprising functionality of IBLTs is supported via hashing. In more detail, the original IBLT construction by Goodrich and Mitzenmacher consists of an array $A$ of $m$ cells along with a hash function $h$ mapping keys to $k$ distinct entries in $A$ for a tuneable parameter $k$. Each cell of $A$ has three fields, a \emph{count}, a \emph{keySum} and a \emph{valueSum}. When inserting a key-value pair $(x,y)$, we compute the $k$ positions $h(x) = (i_1,\dots,i_k)$, increment the \emph{count} field in $A[i_j]$, add $x$ to the \emph{keySum} of $A[i_j]$ and add $y$ to the \emph{valueSum} of $A[i_j]$ for each $j=1,\dots,k$. A deletion of a key-value pair is simply supported by reversing these operations, i.e. decrementing \emph{count} and subtracting $x$ from \emph{keySum} and $y$ from \emph{valueSum}. To support the retrieval of the value associated with a query key $x$, we again compute $h(x) = (i_1,\dots,i_k)$ and examine the entries $A[i_j]$. If we find such an entry where the \emph{count} field is one, then we know that only one key-value pair hashed there. We can thus compare the \emph{keySum} to $x$, and if they are equal, we can return the \emph{valueSum}. If the \emph{keySum} is different from $x$, or we find a cell with a \emph{count} of zero, we may return that $x$ is not in the IBLT. Finally, if all $k$ \emph{count} fields are at least two, we return ``Don't know''. If the number of cells $m$ is $2nk$, then the chance that a key-value pair hashes to at least one unique entry (no collisions) is around $1-2^{-\Omega(k)}$ whenever the number of key-value pairs stored in the IBLT does not exceed the threshold $n$.

\paragraph*{Peeling.}
The simple functionality above supports Insertions, Deletions and Get operations, where a Get operation retrieves the value associated with a query key $x$. Using space $O(nk)$, the Get operation succeeds with probability $1-2^{-\Omega(k)}$. However, in several applications, such as set reconciliation, one is more interested in outputting the list of all key-value pairs present in the IBLT. For this purpose, a ListEntries operation is also supported. To list all key-value pairs in the IBLT, we repeatedly look for a cell in $A$ with a \emph{count} of one. When we find such a cell $A[i]$, we output $(x,y) = $($A[i]$.\emph{keySum}, $A[i]$.\emph{valueSum}) and then delete $(x,y)$ from the IBLT. This process of \emph{peeling} the key-value pairs reduces the \emph{count} of other fields and thus increases the chance that we can continue peeling key-value pairs. Concretely, the ListEntries operation can be shown to succeed with probability $1-\Omega(n^{-k + 2})$ when the number of key-value pairs present in the IBLT does not exceed the threshold $n$. The peeling success probability thus far exceeds that of the simple Get operation when hashing to at least $k=3$ entries.

\paragraph*{Supporting False Deletions.}
The attentive reader may have observed that the simple version of the IBLT described above critically assumes that no deletions are performed on key-value pairs that are not already present in the IBLT. In the set reconciliation example, this is insufficient as there may be key-value pairs in $S_B$ that are not in $S_A$, which will cause false deletions. A simple extension to the IBLT ensures that it also functions if the total number of present key-value pairs plus the number of false deletions does not exceed the threshold $n$. For set reconciliation, this is equivalent to $|S_A \setminus S_B| + |S_B \setminus S_A| \leq n$. To support such false deletions, we add a \emph{hashSum} field to every cell and include another hash function $g$ mapping keys to a sufficiently large output domain $[R]$. When inserting key-value pairs, $g(x)$ is added to the \emph{hashSum} field of $A[i_j]$ and subtracted during deletions. To retrieve the value associated with a key $x$, we proceed as before, but whenever the \emph{count} is either $-1$ or $1$, we also perform a check that the \emph{hashSum}  is equal to $g$ applied to the \emph{keySum}. If not, we treat the cell as if the count was at least $2$. For ListEntries, a peeling operation also includes such checks and furthermore, when a \emph{count} is $-1$, we may instead insert $(x,y)=$(-\emph{keySum},-\emph{valueSum}) if $g$ applied to -\emph{keySum} equals -\emph{hashSum}. A second source of error is when the same key has been inserted with multiple different values. We ignore this issue here, and remark that the ListEntries in the original IBLT also fails in recovering keys with multiple associated values.

\paragraph*{Memory Usage and Randomness.}
In this paper, we focus on the more interesting ListEntries operation and ignore the Get operation. Requiring that ListEntries succeeds with probability $1-\delta$, the classic IBLT uses space $O(n(\lg(1/\delta)/\lg n + 1))$, since we must set $k = O(1 + \lg_n(1/\delta))$ to make $n^{-k+1} \leq \delta$, and the space usage is $m=O(nk)$ cells. Notice here, and throughout the paper, that space is measured in number of \emph{cells} of the IBLT. In terms of bit complexity, the \emph{count} field needs $O(\lg n)$ bits, the \emph{keySum} and \emph{valueSum} fields need $O(\lg |U| + \lg n)$ bits when keys and values come from a universe $U$. Finally, in both previous IBLTs and our new construction, the \emph{hashSum} field needs $O(\lg(1/\delta) + \lg n)$ bits. Thus each cell of the table costs $O(\lg(|U|n/\delta))$ bits. 

The analysis of the classic IBLT critically assumes that the hash function $h$ is truly fully random. This is of course unrealistic in practice. But where many typical data structures can make due with $O(\lg(1/\delta))$ or $O(\lg n)$-wise independent hash functions, this is not known to be the case for the IBLT. Concretely, the standard analysis of the peeling process of the IBLT requires a union bound over exponentially many events (for every set of $2 \leq j \leq n$ keys $S$, for every set $T$ of $jk/2$ entries of $A$, we have a failure event saying that $h(x) \in T$ for all $x \in S$). With exponentially many events in the union bound, each of them must occur with probability at most $\exp(-\Omega(n))$ for the union bound to be useful. This requires a seed length of $\Omega(n)$ bits for a hash function and thus cannot be implemented with $k$-wise independence for $k$ significantly less than $n$. It could be the case that a more refined analysis could show that less randomness suffices, but this has not yet been demonstrated. 

We remark that it is possible to show that tabulation hashing~\cite{FOCS:DKRT15,Thorup17} may be used to support peeling, but this also requires a random seed of length proportional to $n$, since it requires a character size of at least $(1+\Omega(1))n$, and the space usage is at least the number of characters.
Finally, we mention that it may also be possible to use the splitting trick of Dietzfelbinger and Rink~\cite{ICALP:DieRin09}, but as far as we are aware of, it would be not more efficient than tabulation hashing in this context.

\subsection{Our Contributions}
Our main contribution is a new version of the IBLT that is both more space efficient and that can be implemented with much less randomness. We call our new data structure a Stacked IBLT and show the following:
\begin{theorem}
Let $\delta$ be less than a sufficiently small constant.
Given a threshold $n$, the Stacked IBLT supports Insertions, Deletions and ListEntries operations, where ListEntries succeeds with probability $1-\delta$ when the number of key-value pairs is no more than $n$. Furthermore, it uses space $O(n + \lg(1/\delta)\lglg(1/\delta))$ cells and requires only $O(\lg(\lg(n)/\delta))$-wise independent hashing.
\end{theorem}
Comparing this to the classic IBLT, our construction outperforms it for any $\delta = n^{-\omega(1)}$ and more importantly, it can be implemented with a small random seed. Our Stacked IBLT also supports false deletions like the classic IBLT and ListEntries succeeds with the claimed probability if the number of key-value pairs plus the number of false deletions does not exceed $n$.

We note that such small failure probabilities are important in cryptographic applications, like the ones that rely on encrypted compression~\cite{CCS:CDGLY21,C:LiuTro22}.
A data-dependent failure of a data structure leaks information about its contents, even if one can not see the contents of the data structure itself.
In cryptographic applications, where security should commonly break with at most a negligible probability, using a (encrypted) data structure, which fails with an inverse polynomial probability is insufficient.
An adversary could deduce information about encrypted data by just observing, whether a cryptographic protocol successfully terminates or not. 

The overall idea in the Stacked IBLT is to construct arrays $A_1,\dots,A_{\lg n}$ where $A_i$ has $C n/2^i$ entries. Each of the arrays has its own hash function $h_i$ mapping keys to a single entry in $A_i$. To support the ListEntries operation, we start by peeling all elements in  $A_1$ that hash uniquely. We then proceed to $A_2$ and so forth. The critical property we require is that each time we peel, we successfully peel at least half of all remaining key-value pairs. In this way, the number of entries in the next $A_i$ to peel from, is always a constant factor larger than the number of remaining key-value pairs. When we reach $A_{\lg n}$, we finally peel the last key-value pair. In this way, all we need from the hash functions $h_i$, is that at least half the key-value pairs hash uniquely with probability $1-\delta/\lg n$. We prove that this is the case if the $h_i$'s are just $O(\lg(\lg(n)/\delta))$-wise independent:
\begin{theorem}
  \label{thm:limitind}
Let $x_1,\dots,x_n \in U$ be a set of $n$ distinct keys from a
universe $U$ and let $h : U \to [Cn]$ be a hash function drawn from a
$2k$-wise independent family of hash functions.  If $C \geq 8e$, then
with probability at least $1-4 \cdot (8e/C)^{\min\{k,n/C\}}$ it holds that there are no more than $n/2$ indices $i$ such that there exists a $j \neq i$ with $h(x_i) = h(x_j)$.
\end{theorem}
In addition to allowing implementations with limited independence, the geometrically decreasing sizes of the arrays $A_i$ also result in the improved space usage compared to classic IBLTs.

While \autoref{thm:limitind} might at first sight appear to
follow from standard approaches for analyzing hash functions with
limited independence, there are in fact several difficult obstacles that we
need to overcome to prove it. In particular, as $k$ approaches $n$, the obvious approaches fail miserably. Furthermore, our Stacked IBLTs critically needs Theorem~\ref{thm:limitind} to hold for $k$ all the way up to $n$. 
We believe the ideas we use to overcome this barrier are interesting in their own right and may prove useful in future work. 
We thus discuss these ideas and the barriers we overcome in Section~\ref{sec:tech}.

Let us also
comment on the constant $8e$. It is not as small as one could hope,
but it is small enough that we have chosen to state it explicitly rather
than hide it in $O$-notation. Presumably our analysis could be
tightened further to reduce it by a constant factor, but we have focused on a
clean exposition of the proof.

Finally, let us also comment that when the number of remaining key-value pairs drop below $\lg(1/\delta)$, Theorem~\ref{thm:limitind} is insufficient to guarantee a success probability of $1-\delta/\lg n$ due to the $\min\{k,n/c\}$ in the exponent. For this reason, we change strategy and replace some of the arrays $A_i$ by matrices with multiple rows. We leave the details to later sections and mention here that this is what causes the $\bigO{\lg(1/\delta)\lglg (1/\delta)}$ term in the space usage of the Stacked IBLT.

In terms of computational efficiency our construction is slightly worse than that of Goodrich and Mitzenmacher.
Retrieving all key-value pairs from their IBLT has a computational cost of $\bigO{n \cdot (1 + \lg(1/\delta)/\lg(n)}$, while our construction requires $\bigO{n \cdot \lg(n/\delta)}$.
In our opinion, however, this is a small price to pay for achieving smaller IBLTs that require less randomness.

\paragraph*{Encrypted Compression.}
We apply our new data structure to the encrypted compression problem, studied by Fleischhacker, Larsen, and Simkin~\cite{EC:FleLarSim23}.
Here one is given an array of ciphertexts of a homomorphic encryption scheme, where at most $t$ are encryptions of non-zero values.
The goal of an encrypted compression scheme is to compress this vector as much as possible, without knowing what is inside the ciphertexts, i.e. without knowing which entries in the vector are encryptions of zero and which are not.
Apart from being theoretically interesting, this problem also naturally appears as part of larger cryptographic protocols~\cite{CCS:CDGLY21,C:LiuTro22}.
We show that following the approach of Fleischhacker, Larsen, and Simkin one can use our stacked IBLT data structure to obtain better encrypted compression schemes.
Additionally, we show how their approach can be generalized to work for arbitrary homomorphic encryption schemes.
Note that their work, required the encryption schemes to have plaintext spaces that grow at least linearly with the desired upper bound on the error rate of their data structure.
We provide a detailed description of the improved compression scheme in~\autoref{app:encryptedcompression}.

\paragraph*{Rateless IBLTs.}
In a work subsequent to ours, Yang, Gilad, Alizadeh~\cite{YanGilAli24} consider the setting of rateless IBLTs.
Here an encoder has a fixed set of source symbols and would like to encode them into an infinite sequence of coded symbols.
Without going into detail, these coded symbols should have several high-level properties:
The computation of the coded symbols should not depend on a fixed a-priori threshold of how many source symbols will be in the data structure.
The sequence of generated coded symbols should be linear in the sense that two sequences of coded symbols can be subtracted to obtain a sequence of coded symbols that represents the set difference of the corresponding sets.
For any number of source symbols, one should be able to decode them back from a sufficiently long prefix of the sequence of coded symbols.

As noted by Yang, Gilad, Alizadeh, the IBLT of Goodrich and Mitzenmacher~\cite{All:GooMit11} does not satisfy these properties as the size of the data structure needs to be fixed at the start and there is no clear way of viewing it as a infinite sequence of coded symbols.
We will not prove this formally in our work, but note that our stacked IBLTs naturally have these properties, as they can be constructed starting from the smallest array and repeatedly building the larger arrays on top of it, viewing the array cells as coded symbols.

\subsection{Some More Related Works.}
A variant of IBLTs that may appear similar to ours are irregular IBLTs, as originally already suggested by Goodrich and Mitzenmacher~\cite{All:GooMit11} and also studied by L{\'a}zaro and Bal{\'a}zs Matuz~\cite{ISTC:LazMat21}, where different set elements are encoded using a different amount of hash functions. 
We note that our construction is regular, since it is oblivious to the specific value of any one set element and all elements get treated equally. 
We believe this to be helpful for applications, like encrypted compression, where the set elements are not visible to the encoder generating the data structure.

In a recent work, that appeared subsequent to ours, by Belazzougui, Kucherov, Walzer~\cite{ICALP:BelKucWal24}, the authors consider IBLTs with very small failure probabilities as we do here. The idea behind their construction is to augment the original IBLT of Goodrich and Mitzenmacher with a smaller backup stash data structure. When decoding of the main IBLT fails, their peeling resorts to recovering the missing elements from the stash.
In comparison, our stacked IBLTs can conceptually be seen as an iterative version of this idea, as we have a sequence of smaller and smaller ``stashes'', moving on to peeling the smaller ones, when peeling the bigger ones fails repeatedly.
Furthermore, their work considers fully random hash functions, whereas our work gets away with using hash functions with limited independence.
Their construction results in a sketch that is asymptotically comparable in size and has a better expected, but worse worst-case decoding time.

\subsection{Technical Contributions}
\label{sec:tech}
When analysing events involving hash functions of limited independence, one typically considers higher moments of a sum of random variables that each depends only on a constant number of hash values. For our Theorem~\ref{thm:limitind}, the natural random variables to consider would be the random variables $X_{i,j}$ taking the value $1$ if $h(x_i)=h(x_j)$. Clearly there are no more than $n/2$ indices $i$ such that there exists $j \neq i$ with $h(x_i) = h(x_j)$ if $\sum_{i \neq j} X_{i,j} \leq n/2$. To upper bound $\Pr[\sum_{i\neq j} X_{i,j} > n/2]$, we raise both sides of the inequality to the $k$'th power and use that $\Pr[\sum_{i\neq j} X_{i,j} > n/2] = \Pr[(\sum_{i\neq j} X_{i,j})^k > (n/2)^k]$. Using Markov's inequality, this probability is at most $\E[(\sum_{i \neq j} X_{i,j})^k]/(n/2)^k$. Expanding the $k$'th power of the sum into a sum of monomials and using linearity of expectation, we have $\E[(\sum_{i\neq j} X_{i,j})^k] = \sum_{T \in \{(i,j) : i\neq j\}^k} \E[\prod_{(i,j) \in T} X_{i,j} ]$. Since each product depends on at most $2k$ hash values, and $h$ is $2k$-wise independent, we can analyse each monomial as if $h$ was truly random.

For the purpose of proving our theorem, this approach actually suffices to establish the theorem for $k < \sqrt{n}$. However, for our application in IBLTs we need the theorem to hold for $k$ up to $\Omega(n)$. The problem is that as $k$ approaches $n$, using that $\sum_{i\neq j} X_{i,j}$ is small as a proxy for having many elements hash to a unique position is lossy. In essence, this is because $\ell$ elements hashing to the same value contributes around $\ell^2$ to $\sum_{i\neq j} X_{i,j}$ whereas it actually only corresponds to $\ell$ elements not hashing to a unique value. For this reason, $\E[(\sum_{i\neq j} X_{i,j})^k]$ is simply too large to give a meaningful bound from Markov's inequality when $k = \Omega(\sqrt{n})$. In fact, it is not only the higher-moments method that is doomed, but any approach based on arguing that $\Pr[\sum_{i\neq j} X_{i,j} > n/2]$ is small will fail. Consider for instance the case where $k$ is $\Theta(n)$. Our Theorem~\ref{thm:limitind} shows that the probability that less than $n/2$ keys hash uniquely is $\exp(-\Omega(n))$. If we consider $\sum_{i\neq j} X_{i,j}$ and even assume that $h$ is truly random, then the probability that the first $n/\lg n$ keys all hash to the first $n/\lg^3 n$ entries is $(C\lg^3 n)^{-n/\lg n} \geq \exp(-O(n \lglg n/\lg n))$ for constant $C>0$. But when this happens, we have $\sum_{i \neq j} X_{i,j} \geq (n/\lg^3 n) 2 \binom{\lg^2n}{2} \approx n \lg n$. That is, $\Pr[\sum_{i\neq j} X_{i,j} > n/2] \geq \exp(-O(n \lglg n/\lg n))$.

In light of this, it is not a priori clear which random variables 
are sensible to analyse, keeping in mind that they should depend on
only few hash values (for the sake of limited independence) and yet
accurately capture the event that at least $n/2$ elements hash to a
unique value. We present two alternative proofs circumventing this barrier.

In the first, and completely self-contained proof, we carefully
define random variables $Y_{i,j}$ that actually depend \emph{on all}
hash values. We then consider the $k$'th moment of a sum involving
these $Y_{i,j}$'s and argue that most monomials are $0$ due to the special definition of the $Y_{i,j}$'s. Now that there are only very few non-zero monomials left, we upper bound our $Y_{i,j}$'s by the $X_{i,j}$'s above, bringing us back into a setup with monomials depending on at most $2k$ hash values. Compared to going directly from the $X_{i,j}$'s, what we win is that there are much fewer monomials left in the sum. The initial pruning of monomials using the more involved $Y_{i,j}$'s is a key technical innovation that we have not seen before and believe may be an inspiration in future work analysing random variables of limited independence.

In the second proof, we invoke a previous theorem on $k$-wise independence fooling combinatorial rectangles~\cite{rectangles, v016a017}. This proof is shorter than the first, but relies on the heavy lifting done in previous works and does not yield the explicit small constant in our theorem.

%
%


\section{Preliminaries}\label{sec:preliminaries}

Let $X,Y$ be sets, we denote by $\abs{X}$ the size of $X$ and by $X \ssdiff Y$ the symmetric set difference of $X$ and $Y$, i.e., $X \ssdiff Y = (X \cup Y)\setminus (X\cap Y) = (X \setminus Y)\cup (Y\setminus X)$.
We write $x \gets X$ to denote the process of sampling a uniformly random element $x \in X$.
Let $\vec{v} \in X^n$ be a vector.
We write $v_i$ to denote its $i$-th component.
Let $\vec{M} \in X^{n \times m}$ be a matrix.
We write $\vec{M}[i,j]$ to denote the cell in the $i$-th row and $j$-th column.
We write $[n]$ to denote the set $\{1, \dots, n\}$.
We write $\lg$ without a specified base to denote the logarithm to base two.


\section{Hashing Uniquely with Limited Independence}\label{sec:hashing}

In this section, we prove our main technical result, Theorem~\ref{thm:limitind}, which we restate here for convenience.
\begin{customTheorem}{2}[restated]
Let $x_1,\dots,x_n \in U$ be a set of $n$ distinct keys from a
universe $U$ and let $h : U \to [Cn]$ be a hash function drawn from a
$2k$-wise independent family of hash functions.  If $C \geq 8e$, then
with probability at least $1-4 \cdot (8e/C)^{\min\{k,n/C\}}$ it holds that there are no more than $n/2$ indices $i$ such that there exists a $j \neq i$ with $h(x_i) = h(x_j)$.
\end{customTheorem}
As discussed in Section~\ref{sec:tech}, the straight forward approach of analysing moments of a sum $\sum_{i<j} X_{i,j}$ with $X_{i,j}$ being an indicator for $h(x_i)=h(x_j)$, does not give the desired result. In essence, this is because a collision of $\ell$ elements contributes roughly $\ell^2$ to the sum.

In this section, we present two alternative proofs circumventing this barrier. We start by giving the self-contained proof that introduces an elegant new trick to analysing $k$-wise independent random variables. We then give a proof invoking results on $k$-wise independence fooling combinatorial rectangles. The remark that the second proof does not yield the explicit constants in Theorem~\ref{thm:limitind}.

\subsection{Proof via Moments}
Our first step in the proof of Theorem~\ref{thm:limitind} is thus to make a far less obvious definition of random variables.

\begin{proof}
Define random variables $Y_{i,j}$ with $i\neq j$ taking the value $1$ if $h(x_i) = h(x_j)$ and furthermore, for all $a$ with $\min\{i,j\} < a < \max\{i,j\}$ we have $h(x_i)\neq h(x_a)$. Otherwise, $Y_{i,j}$ takes the value $0$. Observe that if elements $x_{i_1},\dots,x_{i_\ell}$ are all those that hash to a concrete value $v$, and $i_1 < i_2 < \cdots < i_\ell$, then $Y_{i_1,i_2}=Y_{i_2,i_1}=Y_{i_2,i_3}=\cdots = Y_{i_\ell, i_\ell-1} = 1$ and all other $Y_{i,j}$'s with $i$ or $j$ in $\{i_1,\dots,i_\ell\}$ are zero. The random variable $Y_{i,j}$ is thus $1$ if $x_i$ and $x_j$ hash to the same $v$, and furthermore, $i$ and $j$ are consecutive in the sorted order of all elements hashing to $v$. Critically, a collision of $\ell$ elements contribute only $2\ell-2$ to $\sum_{i\neq j} Y_{i,j}$. On the negative side, these random variables $Y_{i,j}$ clearly depend on more than two hash values unlike the $X_{i,j}$'s.

Letting $S = \{x_1,\dots,x_n\}$, observe that if there more than $n/2$ keys $x \in S$ such that there is a $y \in S \setminus\{x\}$ with $h(x)=h(y)$, then $\sum_{i \neq j} Y_{i,j} > n/2$.  Let $r = \min\{k,n/C\}$. Using Markov's, we get
\begin{eqnarray}
  \label{eq:markov}
  \Pr\left[\sum_{i \neq j} Y_{i,j} > n/2\right] = \Pr\left[\left(\sum_{i \neq j} Y_{i,j}\right)^r > (n/2)^r \right] < \frac{\E\left[\left(\sum_{i \neq j} Y_{i,j}\right)^r\right]}{(n/2)^r}.
\end{eqnarray}
We thus focus on bounding $\E[(\sum_{i \neq j} Y_{i,j})^r]$. Expand it into its monomials
 \[
    \E\left[\left(\sum_{i \neq j} Y_{i,j} \right)^r \right] = \sum_{(i_1,j_1),\dots,(i_r,j_r)} \E \left[ \prod_{h=1}^r Y_{i_h, j_h} \right].
  \]
Here the sum ranges over all lists of $r$ pairs $(i_h,j_h)$ with $i_h \neq j_h$. Notice that the product is $1$ if and only if all the indicators involved are $1$. For a monomial $\prod_{h=1}^r Y_{i_h,j_h}$, think of the pairs $(i_h,j_h)$ as edges of a graph with the elements $x_1,\dots,x_n$ as nodes. The critical observation is that if any node in this graph has at least three distinct neighbors, then $\prod_{h=1}^r Y_{i_h, j_h}=0$. To see this, assume the node $x_i$ has at least three distinct neighbors. If $x_i$ has two neighbors $x_{j_1},x_{j_2}$ with $j_1 < j_2 < i$, then we cannot have both $Y_{j_1, i} = Y_{i,j_1} = 1$ and $Y_{j_2,i} = Y_{i,j_2}=1$. This is because, by definition, $Y_{j_1,i}$ can only be $1$ if there are no elements $x_a$ with $h(x_a) = h(x_{j_1})$ and $j_1 < a<i$. But $a=j_2$ is an example of such an element when we also require $Y_{j_2,i}=Y_{i,j_2}=1$. A similar argument applies to the case that $x_i$ has two neighbors $x_{j_1},x_{j_2}$ with $i < j_1 < j_2$. Notice that this also implies that the monomial is $0$ if the corresponding graph has a cycle since the node of largest index on the cycle has an edge to two distinct neighbors of lower index. In combination, the monomial can only be non-zero if the corresponding edges form connected components corresponding to paths (possibly with duplicate edges).

Let $\cG^r$ denote the set of all ordered lists $\cL$ of $r$ pairs $\cL := (i_1,j_1),\dots,(i_r,j_r)$ (with $i_h \neq j_h$ for all $h$) such that every connected component in the corresponding graph $G(\cL)$ forms a path. Then 
\[
  \E\left[\left(\sum_{i \neq j} Y_{i,j} \right)^r \right] = \sum_{\cL \in \cG^r} \E \left[ \prod_{(i,j) \in \cL} Y_{i,j} \right].
  \]
  Now consider a monomial $\prod_{(i,j) \in \cL} Y_{i, j}$ for an $\cL \in \cG^r$. Define $X_{i,j}$ as the random variable taking the value $1$ if $h(x_i) = h(x_j)$ and $0$ otherwise. Here we use that $Y_{i,j} \leq X_{i,j}$ and thus $\prod_{(i,j) \in \cL} Y_{i, j}  \leq \prod_{(i,j)\in \cL} X_{i, j}$. Therefore
    \[
  \E\left[\left(\sum_{i \neq j} Y_{i,j} \right)^r \right] \leq \sum_{\cL \in \cG^r} \E \left[ \prod_{(i,j) \in \cL} X_{i,j} \right].
  \]
What we have achieved is to upper bound $\E[(\sum_{i < j} Y_{i,j})^r]$ by the contribution from monomials corresponding to graphs consisting of paths. Furthermore, for these monomials, we have replaced the $Y_{i,j}$ variables by the simpler $X_{i,j}$ variables that each only depend on two hash values. This allows us to handle the limited independence of $h$.

Next, we bound $\E[ \prod_{(i,j) \in \cL} X_{i,j}]$ for an $\cL \in \cG^r$. With the graph interpretation $G(\cL)$ of $\cL$ in mind, we observe that the product is $1$ if and only if, for every connected component in $G(\cL)$, all nodes in the component hash to the same value. Furthermore, the monomial depends on at most $2r \leq 2k$ hash values and thus the random variables behave as if $h$ was truly random. For a connected component with $q_i$ nodes, the probability all nodes hash to the same is precisely $(Cn)^{-(q_i - 1)}$. If the total number of nodes in $G(\cL)$ having at least one neighbor is $q$ and the total number of connected components in $G(\cL)$ formed by these nodes and their edges is $c$, then
  \[
    \E \left[ \prod_{(i,j) \in \cL}X_{i, j} \right] = (Cn)^{-q +c}.
  \]
  For every $q \leq 2r$ and every $c \leq q/2$, let $\cG^r_{q,c} \subseteq \cG^r$ be the subset of lists $\cL$ for which the corresponding graph $G(\cL)$ has $c$ non-singleton connected components and those connected components together have $q$ nodes. Then
  \[
    \E\left[\left(\sum_{i \neq j} Y_{i,j} \right)^r \right] \leq \sum_{q=2}^{2r} \sum_{c =1}^{q/2}  \sum_{\cL \in \cG^r_{q,c}} \E \left[ \prod_{(i,j) \in \cL} X_{i,j} \right] = \sum_{q=2}^{2r} \sum_{c =1}^{q/2}  |\cG^r_{q,c}|(Cn)^{-q+c}.
  \]
  We thus need to bound $|\cG^r_{q,c}|$. Here we show the following
  \begin{lemma}
    \label{lem:count}
    For all $q \leq 2r$, $c \leq q/2$ it holds that
    \[
      |\cG^r_{q,c}| \leq \left(\frac{4er}{q}\right)^{q-c} 2^r q^r n^q q^{-c}.
    \]
  \end{lemma}
  Before we prove the lemma, let us use to finish our proof of Theorem~\ref{thm:limitind}. Continuing our calculations above using Lemma~\ref{lem:count}, we have that
  \begin{eqnarray*}
    |\cG^r_{q,c}|(Cn)^{-q+c} &\leq&\left(\frac{4er}{qCn}\right)^{q-c} 2^r q^r n^q q^{-c}\\
    &=& \left(\frac{4er}{qC}\right)^{q} \left(\frac{4er}{Cn}\right)^{-c} 2^r q^r.
  \end{eqnarray*}
  Since we set $r = \min\{k, n/C\}$ and require $C \geq 8e$, we have $(4er/(Cn)) \leq 1/4$ and thus exploiting that the sum over $c$ is a geometric series we get
  \[
    \sum_{c=1}^{q/2}  |\cG^r_{q,c}|(Cn)^{-q+c} \leq 2 \left(\frac{4er}{qC}\right)^{q} \left(\frac{4er}{Cn}\right)^{-q/2} 2^r q^r = 2 \left( \frac{4e r n}{Cq^2} \right)^{q/2}2^r q^r
  \]
  Using again that $n/C \geq r$ and $r \geq q/2$, we have $4er n/(Cq^2) \geq 4e r^2/q^2 \geq e$ and thus we may again use a geometric series to conclude
  \[
    \E\left[\left(\sum_{i \neq j} Y_{i,j} \right)^r \right]  \leq \sum_{q=2}^{2r} \sum_{c=1}^{q/2}  |\cG^r_{q,c}|(Cn)^{-q+c} \leq 4 \left( \frac{4e r n}{C(2r)^2} \right)^{r}(4r)^r = 4 \left(\frac{4e n}{C}\right)^r.
  \]
  Plugging this back into the bound~\eqref{eq:markov} we got from Markov's inequality, we finally conclude
  \[
    \Pr\left[\sum_{i \neq j} Y_{i,j} > n/2\right] \leq 4 \cdot \left(\frac{8e}{C}\right)^r.
  \]
  Recalling that $r = \min\{k,n/C\}$ completes the proof.
\end{proof}
  
\paragraph*{Counting Graphs (Proof of Lemma~\ref{lem:count}).}
To bound $|\cG^r_{q,c}|$, we first recall that every $\cL \in \cG^r_{q,c}$ corresponds to a graph consisting of $c$ non-singleton connected components, each forming a path of $q_i$ nodes with $q = \sum_i q_i$. The set of (undirected) edges in $G(\cL)$ thus has cardinality $q-c \leq r$. We now argue that any $\cL \in \cG^r_{q,c}$ can be uniquely described by an element in
\[
  \cU := \binom{r}{q-c} \times (\{0,1\} \times [q-c])^{r-(q-c)} \times \binom{2(q-c)}{q} \times [n]^q \times [q]^{2(q-c)-q}.
\]
Here $\binom{r}{q-c}$ is the set of all $(q-c)$-sized subsets of a universe of cardinality $r$. Notice that this indirectly specifies a surjective function from $\cU$ to $\cG^r_{q,c}$ and thus
\[
  |\cG^r_{q,c}| \leq \binom{r}{q-c} (2(q-c))^{r-(q-c)} \binom{2(q-c)}{q}  n^q q^{q-2c}.
\]
To describe an $\cL \in \cG^r_{q,c}$ with an element from $\cU$, use an element in $\binom{r}{q-c}$ to specify the first occurence of each edge in $\cL$ (where an edge $(i,j)$ is first if neither $(i,j)$ or $(j,i)$ occurs earlier in $\cL$). For each of the $r-(q-c)$ remaining edges in order, use an element in $\{0,1\} \times [q-c]$ to specify it as a copy of one of the $q-c$ first edges, where $\{0,1\}$ indicates whether to reverse the order of the end points. Next observe that the $q-c$ first edges have $2(q-c)$ end points of which precisely $q$ are unique. Specify the first occurence of each unique node on these edges using an element in $\binom{2(q-c)}{q}$. Next use an element in $[n]$ for each such node in order to specify it among the nodes $x_1,\dots,x_n$. Finally, for the remaining $2(q-c)-q$ end points, specify them as an index into the $q$ first occurrences of unique nodes. This information uniquely describes $\cL$.

Using that $\binom{2(q-c)}{q}  \leq 2^{2(q-c)}$ and the general inequality $\binom{r}{q-c} \leq (er/(q-c))^{q-c}$, we conclude
\begin{eqnarray*}
  |\cG^r_{q,c}| &\leq& \left(\frac{er}{q-c} \right)^{q-c} (2(q-c))^{r-(q-c)} 2^{2(q-c)} n^q q^{q-2c} \\
                &\leq& \left(\frac{2er}{q} \right)^{q-c} (2q)^{r-(q-c)} 2^{2(q-c)} n^q q^{q-2c} \\
  &=& \left(\frac{4er}{q}\right)^{q-c} 2^r q^r n^q q^{-c}.
\end{eqnarray*}
\qed

Let us finish by commenting on our choice of bounding $\sum_{i \neq j} Y_{i,j}$ rather than $\sum_{i < j} Y_{i,j}$. This choice was made for simplicity, but one may wonder whether focusing on the latter might result in tighter constants. This does not seem to be the case, as then the assumption that there are more than $n/2$ keys $x \in S$ such that there is a $y \in S\setminus\{x\}$ with $h(x)=h(y)$, does not imply $\sum_{i<j} Y_{i,j} > n/2$ (we use $\sum_{i \neq j} Y_{i,j} > n/2$), but only $\sum_{i<j} Y_{i,j} > n/4$. We would thus lose a constant factor in Markov's.


\section{Proof via $k$-Wise Independence Fools Combinatorial Rectangles}\label{app:rectanglesproof}
We now give a second proof based on $k$-wise independence fooling combinatorial rectangles. This proof was communicated to us by an anonymous reviewer.

We first introduce the notion of a combinatorial rectangle. A combinatorial rectangle is a function $f : [m]^n \to \{0,1\}$ which is specified by $n$ coordinate functions $f_i : [m] \to \{0,1\}$ as $f(x_1,\dots,x_n) = \prod_{i \in m} f_i(x_i)$.
We now use the following result, typically attributed to~\cite{rectangles}, although we cannot directly find this statement in the version available online. A clean introduction to combinatorial rectangles and bounded independence can, for instance, be found in~\cite{v016a017}.
\begin{theorem}
  \label{thm:fool}
Let $X_1,\dots,X_n$ be $k$-wise independent random variables with uniform marginal distributions over $[m]$. Then there is a constant $a>0$ such that
  \[
    \left|\E_{X_1,\dots,X_n}\left[f(X_1,\dots,X_n)\right] - \E_{x \in [m]^n}\left[f(x)\right]\right| \leq e^{-a k},
  \]
  where $\E_{x \in [m]^n}$ denotes a uniform random $x \in [m]^n$.
\end{theorem}
With this tool in place, we now prove Theorem~\ref{thm:limitind}.
\begin{proof}
  Recall that we are hashing into $Cn$ bins. Let $x_1,\dots,x_n \in U$ denote the $n$ keys and let $h : U \to [Cn]$ denote a hash function drawn randomly from a $2k$-wise independent family of hash functions. Let $X_i$ be the random variable taking the value $h(x_i)$.

  Let $J \subseteq [Cn]$ be the indices of a subset of the bins, with $|J|=t$ for a parameter $t$ to be determined. Define random variables $Z_j$ taking the value $1$ if no element hashes to the value $j$ and $0$ otherwise. The probability that all bins indexed by $J$ are empty is $\E[\prod_{j \in J} Z_j]$. If we now define functions $f_i : [Cn] \to \{0,1\}$ taking the value $1$ on $x \notin J$ and the value $0$ for $x \in J$, we have that $\prod_{j \in J} Z_j = \prod_{i=1}^n f_i(X_i)$, i.e. $\prod_{j \in J} Z_j$ is in effect a combinatorial rectangle. By Theorem~\ref{thm:fool}, we have
  \[
    \E\left[\prod_{j \in J} Z_j\right] \leq \E_{x \in [Cn]^n}\left[f(x)\right] + e^{-ak}.
  \]
  But $\E_{x \in [Cn]^n}[f(x)] = (1-t/Cn)^n \leq e^{-t/C}$. We now require $t < aCk$ and conclude $\E[\prod_{j \in J} Z_j] \leq 2e^{-t/C}$.

  Next, observe that if there are less than $n/2$ elements that hash to a unique value, then the number of occupied bins is at most $3n/4$. Vice versa, the number of unoccupied bins is at least $Cn-3n/4$. If we also have $t < Cn-3n/4$, then we may bound the expected number of $t$-sized subsets of bins that are empty. That is, if we let $Y_J = \prod_{j \in J} Z_j$, then we have just shown
  \[
    \E\left[\sum_{J \in \binom{Cn}{t}} Y_J\right] \leq \binom{Cn}{t} 2e^{-t/C}.
  \]
  On the other hand, we may also lower bound the expectation by
  \[
     \E\left[\sum_{J \in \binom{Cn}{n}} Y_J\right]  \geq \Pr\left[\sum_j Z_j > Cn-3n/4\right] \binom{Cn-3n/4}{t}.
   \]
   Combining the two yields
   \begin{eqnarray*}
     \Pr\left[\sum_j Z_j > Cn-3n/4\right]  &\leq& 2e^{-t/C}\cdot \frac{ \binom{Cn}{t}}{\binom{Cn-3n/4}{t}} \\
                                           &\leq& 2 \cdot \left(\frac{e^{-1/C}(Cn-t)}{Cn-3n/4-t} \right)^t \\
                                           &=& 2 \cdot \left(e^{-1/C}\left(1 + \frac{3}{4 (C-3/4-t/n)} \right)\right)^t \\
     &\leq& 2 \cdot \left(e^{-1/C + (3/4) \cdot 1/(C-3/4-t/n)}\right)^t
   \end{eqnarray*}
   If we require $t < n$ and $C$ at least a sufficiently large constant, then this is $\exp(-\Omega(t/C))$. Setting $t = \min\{aCk, n\}$ completes the proof.
\end{proof}


\section{Smaller IBLTs with Limited Independence}\label{sec:construction}

In this section, we present a new construction of IBLTs, which we call
stacked IBLTs, that is both asymptotically smaller and requires less
randomness (in \autoref{sec:lower} we also argue that the analysis
of the original IBLT cannot be strengthened to give bounds comparable
to our stacked IBLT).

\subsection{Stacked IBLTs}\label{sec:basicconstruction}

In this section we introduce our new Stacked IBLTs that are more space efficient and allow for a lower randomness complexity.
Essentially the construction consists of $\lg n$ stacked smaller IBLTs.
These IBLTs will be decoded in order and each is sized, such that we will be able to prove that it allows decoding at least half the remaining entries.
This means that after decoding all $\lg n$ IBLTs, at most a single element is left to decode which can then be trivially decoded.

\begin{figure}[t]\centering
  \begin{pchstack}[boxed,center]
    \begin{pcvstack}
    \procedure{$\init(\vec{h})$}{
      \pcfor  0 \leq i < \lg(n) - \lg(\tau)\\
      \quad T_i := \basicinit(1, \lceil C n2^{-i}\rceil, \vec{h}_i)\\
      \pcfor  0 \leq i < \lg(\tau)\\
      \quad i' := \lfloor\lg(n) - \lg(\tau)\rfloor +i\\
      \quad T_{i'} := \basicinit(2^{i}, \lceil C \tau 2^{-i}\rceil, \vec{h}_{i'})\\
      \pcreturn (T_0,\dots,T_{\lceil\lg n\rceil-1})
    }
    \pcvspace
    \procedure{$\encode((T_0,\dots,T_{\lceil\lg n\rceil-1}), S, \vec{h})$}{
      \pcfor 0 \leq i < \lceil\lg n\rceil\\
      \quad T_i := \basicencode(S,\vec{h}_i)\\
      \pcreturn (T_0,\dots,T_{\lceil\lg n\rceil-1})
    }
  \end{pcvstack}
  \pchspace
  \begin{pcvstack}
    \procedure{$\decode((T_0,\dots,T_{\lceil\lg n\rceil-1}),\vec{h})$}{
      S' := \emptyset\\
      \pcfor 0 \leq i < \lceil\lg n\rceil\\
      \quad T_i := \basicpeel(T_i,S')\\
      \quad S' := S' \cup \basicdecode(T_i)\\
      \pcreturn S'\\[2.3pt]
    }
    \pcvspace
    \procedure{$\peel((T_0,\dots,T_{\lceil\lg n\rceil-1}), \tilde S, \vec{h})$}{
      \pcfor 0 \leq i < \lceil\lg n\rceil\\
      \quad T_i := \basicpeel(\tilde S,\vec{h}_i)\\
      \pcreturn (T_0,\dots,T_{\lceil\lg n\rceil-1})
    }
  \end{pcvstack}
  \end{pchstack}

  \caption{Our stacked IBLT construction using basic IBLTs as specified in~\autoref{fig:basicfilter} as a building block. We have that $\tau = C_0 \lg(1/\delta)$ for a sufficiently large constant $C_0 > 0$}\label{fig:stackedfilter}
\end{figure}

\begin{figure}\centering
  \begin{pchstack}[boxed,center]
    \begin{pcvstack}
    \procedure{$\basicinit(\rho, \gamma, \vec{h})$}{
      \vec{K} := 0^{\rho\times \gamma}\\
      \vec{V} := 0^{\rho\times \gamma}\\
      \vec{C} := 0^{\rho\times \gamma}\\
      \pcreturn (\vec{K},\vec{V},\vec{C})\\\\[-2pt]
    }\pcvspace
    \procedure{$\basicencode((\vec{K}, \vec{V}, \vec{C}), S, \vec{h})$}{
      \pcforeach (k,v) \in S\\
      \quad \pcforeach i \in [\rho]\\
      \quad \quad j := h_i(k)\\
      \quad \quad \vec{K}[i,j] := \vec{K}[i,j] + k\\
      \quad \quad \vec{V}[i,j] := \vec{V}[i,j] + v\\
      \quad \quad \vec{C}[i,j] := \vec{C}[i,j] + 1\\
      \pcreturn (\vec{K},\vec{V},\vec{C})
    }
    \end{pcvstack}
	\pchspace
	\begin{pcvstack}
	  \procedure{$\basicdecode((\vec{K}, \vec{V}, \vec{C}),\vec{h})$}{
      S' := \emptyset\\
      \pcfor (i,j) \in [\rho]\times[\gamma]\\
      \quad\pcif \vec{C}[i,j]=1\\
      \quad \quad (k,v) := (\vec{K}[i,j],\vec{V}[i,j])\\
      \quad \quad S' := S'\cup \{(k,v)\}\\
      \pcreturn S'
    }
    \pcvspace
    \procedure{$\basicpeel((\vec{K},\vec{V}, \vec{C}), \tilde S, \vec{h})$}{
      \pcforeach (k,v) \in \tilde S\\
      \quad \pcforeach i \in [\rho]\\
      \quad \quad j := h_i(k)\\
      \quad \quad \vec{K}[i,j] := \vec{K}[i,j] - k\\
      \quad \quad \vec{V}[i,j] := \vec{V}[i,j] - v\\
      \quad \quad \vec{C}[i,j] := \vec{C}[i,j] - 1\\
      \pcreturn (\vec{K},\vec{V},\vec{C}, \vec{h})
    }
    \end{pcvstack}
  \end{pchstack}

  \caption{A simplified version of a basic IBLT for key space $\keyspace$ and universe $\universe$.
  Both $\langle\keyspace,+\rangle$ and $\langle\universe,+\rangle$ need to form groups.
  The basic IBLT requires a number of rows $\rho$, a number of columns $\gamma$ and a vector of hash functions $\vec{h} \in \{h : \keyspace \to [\gamma]\}^\rho$ to initialize.}
  \label{fig:basicfilter}
\end{figure}

Let $n$ be the threshold for an IBLT and $\delta > 0$ a desired failure probability. We can think of our Stacked IBLT as consisting of multiple rows, with a $k$-wise independent hash function associated with each row for $k = \Theta(\lg(\lg(n)/\delta))$. An element is hashed into one position in each row and stored there, like in the classic IBLT.
The key novelty of our solution is that the number of entries per row varies.
Moreover, while a classic IBLT focuses on peeling all elements, our analysis is based on peeling a constant fraction of the elements from each row.

More formally, let $\tau = C_0 \lg(1/\delta)$ for a sufficiently large constant $C_0 > 0$ and assume first that $n \geq \tau$. For $i = 0,\dots,\lg(n/\tau)$, our IBLT has one row $R_i$ with $C n 2^{-i}$ entries. Here $C > 0$ is a sufficiently large constant, where $C=8e$ is provably sufficient. Finally, for $i=0,\dots,\lg(\tau)$, it has a group $G_i$ consisting of $2^i$ rows all with $C \tau 2^{-i}$ entries.
In case $n < C_0 \lg(1/\delta)$, our structure has a group $G_i$ of $2^i$ rows for every $i=\lg(\tau/n),\dots,\lg(\tau)$. In the group $G_i$, every row has $C \tau 2^{-i}$ entries.
The IBLT uses $\sum_{i=0}^{\lg(n/\tau)} C n 2^{-i} + \sum_{i=0}^{\lg(\tau)} C \tau = O(n + \lg(1/\delta)\lglg(1/\delta))$ space.
In the formal description of our Stacked IBLT construction, shown in~\autoref{fig:stackedfilter}, we do not explicitly distinguish between the rows $R_i$ and groups $G_i$, but rather view them as smaller IBLTs that we call $T_1, \dots, T_{\lg{n}}$.
For the analysis, however, distinguishing the smaller IBLTs with one row and those with multiple rows is helpful.

\begin{customTheorem}{1}[restated]\label{thm:mainiblt}
Given a threshold $n$, the Stacked IBLT supports $\encode$, $\peel$, and $\decode$ operations, where $\decode$ succeeds with probability $1-\delta$ if the number of key-value pairs is no more than $n$. Furthermore, it uses space $O(n + \lg(1/\delta)\lglg(1/\delta))$ and requires only $O(\lg(\lg(n)/\delta))$-wise independent hashing.
\end{customTheorem}
\begin{remark}
We note that a $k$-wise independent hash function from a universe $U$ to a universe of size $\gamma$ requires $O(k \lg(U))$ bits.
Since we require $O(\lg(n/\delta)$ such functions, we observe that the total number of random bits we need is $O(\lg(n/\delta)(\lg(1/\delta) + \lg\lg(n)) \lg(U))$ bits.
Regarding running times, the Insert and Delete operations both require $O(\lg(n/\delta))$ evaluations of a $O(\lg(\lg(n)/\delta))$-wise independent hash function, plus insertions in the table entries. The running time is dominated by the evaluations of the hash functions, for a total time of $O(k \lg(n/\delta))$ per element. 
 
\end{remark}
\begin{proof}[Proof of Theorem~\ref{thm:mainiblt}]
To analyse the probability that peeling succeeds, we focus on the case of $n \geq \tau$. The other case is just a special case.

To argue that peeling succeeds with high probability, we consider a very restrictive form of peeling and argue that even this process succeeds. Concretely, for $i=0,\dots,\lg(n/\tau)$, consider peeling all elements that land alone in $R_i$ (after having peeled elements landing alone in $R_j$ with $j < i$). Then, for $i=0,\dots,\lg(\tau)$ in turn, select the row of $G_i$ where most elements hash alone and peel those elements.
To prove that this process succeeds in peeling all elements with probability at least $1-\delta$, we define the events $E_i$ occuring if there are more than $n2^{-(i+1)}$ elements left after peeling from $R_0,\dots,R_i$. Similarly, define $F_i$ as the event that more than $\tau 2^{-(i+1)}$ elements remain after peeling from $R_0,\dots,R_{\lg(\tau)}, G_0,\dots,G_i$. We observe that if $F_{\lg(\tau)}$ does not occur, then there are no more than $1/2$ elements left, i.e. peeling succeeded.

The key step in our proof is to argue that the following two
inequalities hold:
\begin{eqnarray}
  \label{eq:peelsingle}
  \Pr[E_i \mid \cap_{j=0}^{i-1}\overline{E_j}] \leq \frac{\delta}{4 (\lg(n/\tau) - i+1)^2}.
\end{eqnarray}
and
\begin{eqnarray}
  \label{eq:peelmulti}
  \Pr[F_i \mid \cap_{j=0}^{\lg(n/\tau)} \overline{ E_j} \cap_{j=0}^{i-1}\overline{F_j}] \leq \delta^2/2.
\end{eqnarray}
Observe that these two are sufficient as
\begin{eqnarray*}
  \Pr[\overline{F_{\lg(\tau)}}] &\geq& \Pr[\cap_{j=0}^{\lg(n/\tau)} \overline{E_j} \cap_{j=0}^{\lg(\tau)}\overline{F_{j}}] \\
                                  &=& \prod_{i=0}^{\lg(n/\tau)} (1-\Pr[E_i \mid \cap_{j=0}^{i-1} \overline{E_j}]) \prod_{i=0}^{\lg(\tau)} (1-\Pr[F_i \mid \cap_{j=0}^{\lg(n/\tau)} \overline{ E_j} \cap_{j=0}^{i-1}\overline{F_j}]) \\
                                  &\geq& \prod_{i=0}^{\lg(n/\tau)}\left(1-\frac{\delta}{4(\lg(n/\tau) - i+1)^2} \right)\left(1-\delta^2/2\right)^{\lg(\tau)+1} \\
                                  &\geq& 1 - \sum_{i=0}^{\lg(n/\tau)} \frac{\delta}{4(i+1)^2} - \frac{(\lg(\tau)+1)\delta^2}{2} \\
                                  &\geq& 1 - \frac{\delta \pi^2}{24} - \frac{\delta}{2} \\
  &\geq& 1-\delta.
\end{eqnarray*}
We start by showing~\eqref{eq:peelsingle}. Observe that conditioned on $\cap_{j=0}^{i-1}\overline{E_j}$, we know that no more than $n2^{-i}$ elements remain after peeling from $R_0,\dots,R_{i-1}$. We may condition on an arbitrary such set as the hash functions across the rows are independent. So let $S$ be a set of at most $n2^{-i}$ elements. The probability that there are more than $n2^{-(i+1)}$ elements that do no hash alone in $R_i$ is clearly maximized when $|S|$ is $n2^{-i}$. Theorem~\ref{thm:limitind} gives us that this probability is at most $4 (8e/C)^{\min\{k/2, n2^{-i}/C\}}$. For $C \geq 16e$, this is at most $4 \cdot 2^{- \min\{k/2,n2^{-i}/C\}}$. Since $k = \Theta(\lg(\lg(n)/\delta))$, we have $2^{-k/2} < \delta/(4 \lg^2_2 n) \leq \delta/(4(\lg(n/\tau)-i+1)^2)$ for a big enough constant in the $\Theta$-notation. We also have $n2^{-i}/C = \tau 2^{\lg(n/\tau)-i}/C$. For big enough constant $C_0$ (in the definition of $\tau$), this is at least $2 \lg(1/\delta) (\lg(n/\tau)-i+1) + 2 \geq \lg(1/\delta) + 2\lg(\lg(n/\tau)-i+1))+2$ (and this is by a large margin) and we conclude $2^{-n2^{-i}/C} \leq (\delta/4)/(\lg(n/\tau) -i+1))^2$.

To show~\eqref{eq:peelmulti}, note again that conditioned on $\cap_{j=0}^{\lg(n/\tau)} \overline{E_j} \cap_{j=0}^{i-1} \overline{F_j}$, there are at most $\tau 2^{-i}$ elements left after peeling from $R_0,\dots,R_{\lg(n/\tau)}, G_0,\dots,G_{i-1}$. Again, condition on an arbitrary set $S$ of remaining elements. The probability of $F_i$ is clearly maximized if $|S| = \tau 2^{-i}$. We split the proof in two cases. First, assume $\tau 2^{-i} \geq 4 C$. Since each of the $2^i$ rows of $G_i$ have $C \tau 2^{-i}$ entries, and the rows have independent hash functions, it follows by Theorem~\ref{thm:limitind} and $C \geq 16e$, that
\[
  \Pr[F_i \mid  \cap_{j=0}^{\lg(n/\tau)} \overline{E_j} \cap_{j=0}^{i-1} \overline{F_j}] \leq \left(4 \cdot 2^{-\min\{k, \tau 2^{-i}/C\}}\right)^{2^i} \leq \left(2^{-\min\{k/2, \tau 2^{-i}/(2C)\}}\right)^{2^i} .
\]
Here the last inequality assumes $k = \Theta(\lg(\lg(n)/\delta))$ is at least a sufficiently large constant. We also use $\tau 2^{-i}/C - 2 \geq \tau 2^{-i}/C - \tau 2^{-i}/(2C)$. We clearly have $2^{-k/2} \leq \delta^2/2$ for a big enough constant in the $\Theta$-notation. We also have $(2^{-\tau 2^{-i}/(2C)})^{2^i} = 2^{-\tau/(2C)}$. This is again smaller than $\delta^2/2$ for big enough constant $C_0$ in the definition of $\tau = C_0 \lg(1/\delta)$. Finally, for the case where $|S|=\tau 2^{-i} < 4C$, we note that one row of $G_i$ has $C|S|$ entries and thus the expected number of elements that collide with another is no more than $|S|^2/(C|S|) = |S|/C$. By Markov's inequality, the probability that more than $|S|/2$ collide is no more than $2/C < 1/2$. By independence of the rows, the chance that peeling fails is at most $2^{-2^i}$. Since $\tau 2^{-i} < 4C$, we have $2^i \geq \tau/(4C) = C_0 \lg(1/\delta)/(4C)$. For $C_0$ a big enough constant, this implies $2^{-2^i} < \delta^2/2$.
\end{proof}

%


\section{Lower Bound on the Size of IBLTs}\label{sec:lower}

The original IBLT analysis by Goodrich and Mitzenmacher~\cite{All:GooMit11} shows that using truly random hash functions and space $\bigO{nk}$ one can achieve a failure probability of $\bigO{n^{-k+2}}$.
Stated in terms of $\delta$ and $n$, the space usage of their solution is thus $\Omega(n\lg_n(1/\delta))$.
One may wonder, whether their analysis is tight or whether one could prove that IBLTs actually only require $o(nk)$ space for a similar failure probability.

It turns out their space bound is essentially tight and can not be improved by much.
Assume we have an IBLT of size $m$ storing keys $k_1, \dots, k_n$.
Furthermore assume $h_1, \dots, h_k$ are perfectly random hash functions, which map each key to exactly $k$ distinct locations.
For an IBLT to be decodable, we must be able to find a cell with a count of one  at each step of the peeling process.
If $k n \geq c m \lg m$ for some sufficiently large constant $c$, then each cell will have at least $c\lg m$ elements in expectation and thus by Chernoff bound with high probability all cells have a count strictly larger than one.
Thus it must hold that $k n < c m \lg m$.
Consider two distinct keys that are inserted into the IBLT.
The probability that both keys are hashed into exactly the same cells is
\[
\binom{m}{k}^{-1} \geq \left(\frac{em}{k}\right)^{-k} \geq \left( \frac{en}{c\lg m}\right)^{-cm\lg m/n} \geq n^{-cm\lg m/n}.
\]
If we want the IBLT to be correct with probability at least $1 - \delta$, then it has to holds that 
\[
n^{-cm\lg m/n} \leq \delta
\]
and thus
\begin{align*}
&\frac{cm\lg m \lg n}{n} > \lg(1/\delta)\\
\iff &m\lg m > \frac{n\lg(1/\delta)}{c\lg n}.
\end{align*}
For this to hold, it must also hold that
\begin{align*}
&m \lg(n \lg(1/\delta)) > \frac{n\lg(1/\delta)}{c\lg n}\\
\iff & m > \frac{n\lg(1/\delta)}{c\lg(n) \lg(n \lg(1/\delta))}
\end{align*}
and thus it must be true that
\[
m > \frac{n\lg(1/\delta)}{c\lg^2(n \lg(1/\delta))} \geq \frac{n\lg_n(1/\delta)}{c\lg^2(n \lg(1/\delta))}
\]
for any choice of $n \geq 2$.


\section{Supporting Subtraction}\label{sec:supportingsubtraction}
Our IBLT can be made to support such an operation in a manner similar to the original IBLT construction.
As explained previously, we modify the basic IBLT from Section~\ref{sec:basicconstruction} to have an additional hash sum matrix $\vec{H}$ where the values $g(k)$ for keys $k$ for some appropriate hash function $g$ are added up.
During peeling both cells with a count of one or minus one can be peeled, whenever the hash of the key sum cell matches the hash stored in the hash sum cell.
These modification are described in \autoref{fig:stackedfilterwithsubstraction} and \autoref{fig:basicfilterwithsubtract}.
If $g$ is a fully random function, then it is straightforward to see that the modified construction will be correct. 
Using a function $g$ that requires little randomness is slightly more challenging.
We assume that $\keyspace \subseteq \ZZ_p$ for some prime $p$ and we use hash function $g_a(x) = a^x \bmod q$ for some sufficiently large prime $q > p$, which was already used by Mitzenmacher and Pagh~\cite{DC:MP18} in the context of IBLTs.
Such hash functions are useful due to the following lemma.

\begin{figure}\centering
\begin{pcvstack}[boxed]
  \begin{pchstack}
    \procedure{$\init'(\vec{h}\gamechange{,g})$}{
      \pcfor  0 \leq i < \lg(n) - \lg(\tau)\\
      \quad T_i := \gamechange{\basicinit'}(1, \lceil C n2^{-i}\rceil, \vec{h}_i\gamechange{,g})\\
      \pcfor  0 \leq i < \lg(\tau))\\
      \quad i' := \lfloor\lg(n) - \lg(\tau)\rfloor +i\\
      \quad T_{i'} := \gamechange{\basicinit'}(2^{i}, \lceil C \tau 2^{-i}\rceil, \vec{h}_{i'}\gamechange{,g})\\
      \pcreturn (T_0,\dots,T_{\lceil\lg n\rceil-1})
    }\pchspace\!\!\!\!
    \begin{pcvstack}
    \procedure{$\encode'((T_0,\dots,T_{\lceil\lg n\rceil-1}), S, \vec{h}\gamechange{,g})$}{
      \pcfor 0 \leq i < \lceil\lg n\rceil\\
      \quad T_i := \gamechange{\basicencode'}(S,\vec{h}_i\gamechange{,g})\\
      \pcreturn (T_0,\dots,T_{\lceil\lg n\rceil-1})
    }\pcvspace
    \procedure{$\peel((T_0,\dots,T_{\lceil\lg n\rceil-1}), \tilde S, \vec{h}\gamechange{,g})$}{
      \pcfor 0 \leq i < \lceil\lg n\rceil\\
      \quad T_i := \gamechange{\basicpeel'}(\tilde S,\vec{h}_i\gamechange{,g})\\
      \pcreturn (T_0,\dots,T_{\lceil\lg n\rceil-1})
    }
    \end{pcvstack}
	\pchspace\!\!\!\!
    \procedure{$\decode'(\vec{h},g,(T_0,\dots,T_{\lceil\lg n\rceil-1}))$}{
      \gamechange{S_+ := \emptyset, S_- := \emptyset}\\
      \pcfor 0 \leq i < \lceil\lg n\rceil\\
      \quad T_i := \gamechange{\basicpeel'}(T_i ,\gamechange{S_+,S_-},\vec{h}_i\gamechange{,g})\\
      \quad \gamechange{(S_+',S_-') := \basicdecode'(F_i,\vec{h}_i,g)}\!\\
      \quad \gamechange{S_+ := S_+ \cup S_+', S_- := S_- \cup S_-'}\\
      \pcreturn \gamechange{S_+ \cup S_-}
    }
  \end{pchstack}
\end{pcvstack}
  \caption{Modified stacked IBLT supporting subtraction. It makes use of the modified basic IBLT specified in \autoref{fig:basicfilterwithsubtract}.}\label{fig:stackedfilterwithsubstraction}
\end{figure}

\begin{lemma}\label{lem:helpful}
For any $\ell \in \NN$, any $k_1, \dots, k_\ell \in \ZZ_p$, any $\sigma_1, \dots, \sigma_\ell \in \{1, -1\}$, it holds that 
\[
\Pr\left[g_a\Bigl(\sum_{i=1}^{\ell}\sigma_i k_i\Bigr) = \sum_{i=1}^{\ell}\sigma_i g_a\left(k_i\right) \bmod q\right] \leq \frac{2\ell p + 1}{q},
\]
where the probability is taken over the random choice of $a \in \ZZ^*_q$.
\end{lemma}
\begin{proof}
Fix some arbitrary $k_1, \dots, k_\ell \in \ZZ_p$ and $\sigma_1, \dots, \sigma_\ell \in \{1, -1\}$. 
Observe that 
\begin{align*}
&\sum_{i=1}^{\ell} \sigma_i k_i \geq - \ell p \\
\iff &\ell p + \sum_{i=1}^{\ell} \sigma_i k_i \geq 0	
\end{align*}
Next we observe that 
\begin{align*}
&g_a\left(\sum_{i=1}^{\ell}\sigma_i k_i\right) = \sum_{i=1}^{\ell}\sigma_i g_a\left(k_i\right) \bmod q\\
\iff & a^{\sum_{i=1}^{\ell}\sigma_i k_i} = \sum_{i=1}^{\ell}\sigma_i a^{k_i}\bmod q \\
\iff & a^{\ell p + \sum_{i=1}^{\ell}\sigma_i k_i} = a^{\ell p }\sum_{i=1}^{\ell}\sigma_i a^{k_i}\bmod q.
\end{align*}
On the left side of the equation we have a polynomial of degree at most $2\ell p$ with indeterminant $a$.
On the right hand side we have a different polynomial of degree at most $\ell(p + 1)$ with indeterminant $a$.
These polynomials can agree on at most $2\ell p + 1$ points and thus the statement follows. \qed
\end{proof}

\begin{figure}[t]\centering
\begin{pcvstack}[boxed]
  \begin{pchstack}[center]
    \procedure{$\basicinit'(\rho,\gamma,\vec{h},\gamechange{g})$}{
      \vec{K} := 0^{\rho\times \gamma}\\
      \vec{V} := 0^{\rho\times \gamma}\\
      \vec{C} := 0^{\rho\times \gamma}\\
      \gamechange{\vec{H} := 0^{\rho\times \gamma}}\\
      \pcreturn (\vec{K},\vec{V},\vec{C},\gamechange{\vec{H}})
    }
    \pchspace
    \procedure{$\basicdecode'((\vec{K}, \vec{V}, \vec{C}, \gamechange{\vec{H}}),\vec{h},\gamechange{g})$}{
      \gamechange{S_+ := \emptyset, S_- := \emptyset}\\
      \pcfor (i,j) \in [\rho]\times[\gamma]\\
      \quad\pcif \gamechange{\vec{C}[i,j] \in \{1,-1\} \pcand \vec{C}[i,j]\cdot \vec{H}[i,j]=g(\vec{C}[i,j]\cdot\vec{K}[i,j])}\\
      \quad \quad (k,v) := (\gamechange{\vec{C}[i,j]\cdot{}}\vec{K}[i,j],\gamechange{\vec{C}[i,j]\cdot{}}\vec{V}[i,j])\\
      \quad \quad \gamechange{\pcif \vec{C}[i,j] = 1}\\
      \quad \quad \quad \gamechange{S_+ := S_+\cup \{(k,v)\}}\\
      \quad \quad \gamechange{\pcelse}\\
      \quad \quad \quad \gamechange{S_- := S_-\cup \{(k,v)\}}\\
      \pcreturn \gamechange{(S_+,S_-)}
    }
  \end{pchstack}
  \begin{pchstack}[center]
    \procedure{$\basicencode'((\vec{K},\vec{V},\vec{C},\gamechange{\vec{H}}), S,\vec{h}, \gamechange{g})$}{
      \pcforeach (k,v) \in S\\
      \quad \pcforeach i \in [\rho]\\
      \quad \quad j := h_i(k)\\
      \quad \quad \vec{K}[i,j] := \vec{K}[i,j] + k\\
      \quad \quad \vec{V}[i,j] := \vec{V}[i,j] + v\\
      \quad \quad \vec{C}[i,j] := \vec{C}[i,j] + 1\\
      \quad \quad \gamechange{\vec{H}[i,j] := \vec{H}[i,j] + g(k)}\\
      \pcreturn (\vec{K},\vec{V},\vec{C},\gamechange{\vec{H}})
    }
	\pchspace
    \procedure{$\basicpeel'((\vec{K},\vec{V}, \vec{C},\gamechange{\vec{H}), \tilde S_-, \tilde S_+}, \vec{h},\gamechange{g})$}{
      \pcforeach (b,k,v) \in \gamechange{\{(1,k,v) | (k,v) \in \tilde S_+\}\cup\{(-1,k,v) | (k,v) \in \tilde S_-\}}\\
      \quad \pcforeach i \in [\rho]\\
      \quad \quad j := h_i(k)\\
      \quad \quad \vec{K}[i,j] := \vec{K}[i,j] - \gamechange{b\cdot{}} k\\
      \quad \quad \vec{V}[i,j] := \vec{V}[i,j] - \gamechange{b\cdot{}} v\\
      \quad \quad \vec{C}[i,j] := \vec{C}[i,j] - \gamechange{b}\\
      \quad \quad \gamechange{\vec{H}[i,j] := \vec{H}[i,j] - b\cdot g(k)}\\
      \pcreturn (\vec{K},\vec{V},\vec{C},\gamechange{\vec{H}})
    }
  \end{pchstack}

\end{pcvstack}
  \caption{The modified basic IBLT that supports subtraction of IBLTs. This modified IBLT additionally requires a hash function $g : \keyspace \to \ZZ_q$ sampled from the family described above.}
  \label{fig:basicfilterwithsubtract}
\end{figure}

\begin{theorem}
  Let $\vec{h}$ be a vector of functions drawn from appropriate families of $\lg(\lg(n)/\delta)$-wise independent functions and let $g : \ZZ_p \to \ZZ_q$ be chosen uniformly at random as described above for $q \geq \frac{2C n^3 \lg(1/\delta)\lglg(1/\delta) p}{\delta}$ for some sufficiently large constant $C$.
  Then for the modified IBLT described in \autoref{fig:stackedfilterwithsubstraction}, for any pair of sets $S,S' \subseteq \universe$ such that $|S \ssdiff S'| < n$, it holds that
  \[
    \Pr[\decode'(\vec{h},g,\encode'(\vec{h},g,S) - \encode'(\vec{h},g,S')) = S \ssdiff S'] \geq 1 - 2\delta.
  \]
\end{theorem}

\begin{proof}

Note that in our new decoding process, we may have counter entries of one or minus one for cells that contain more than one key.
To see this consider a cell with $k_1 + k_2 - k_3$, where $k_1, k_2, k_3$ are all distinct.
The count is one, but the cell actually still contains three keys.
Storing the sum of hashes of the keys in a cell is intended to prevent mistakenly considering such a cell peelable.
This is the only new source of failure for the decoding algorithm. 
Mistaking a peelable cell as not peelable is not possible.

Recall that an IBLT is a key-value datastructure and thus keys are unique.
That is, every key is inserted into one of the individual IBLTs that we will subtract from each other at most once.
Obviously, a key may still be inserted in both, one, or neither of the two IBLTs.
First, consider an inefficient hash function $\tilde g : \keyspace \to \bin^{|\keyspace|}$ defined as mapping a key $k$ to the bitstring of all zeroes with a single one bit at position $k$.
Note that for sets $X$ and $Y$ of keys, the value
\[
\sum_{x \in X}\tilde g(x) - \sum_{y \in Y}\tilde g(y)
\]
fully encodes the symmetric set difference between $X$ and $Y$.
Thus using this hash function we ensure that no cell is ever peeled incorrectly and we thus obtain the correct output from the decoding procedure.

Let us fix a vector of hash functions $\vec{h}$ and consider two different IBLT decoding runs.
In the first $\tilde g$ is used as the hash function.
In the second one $g_a$ is used.
As long as $g_a$ makes no mistakes, the two peeling processes will behave identically.
Thus to show that decoding works correctly, we simply need to show that peeling using $g_a$ behaves identically to using $\tilde g$.
Let $E_{i, c}$ be the event that a cell $c$ is not peelable after $i$ steps in the decoding process using $\tilde g$, but 
\[
g_a\Bigl(\smashoperator{\sum_{k \in \vec{k}_{i,c}}}\sigma_k k\Bigr) = \smashoperator{\sum_{k \in \vec{k}_{i,c}}}\sigma_k g_a\left(k\right) \bmod q,
\]
where $\vec{k}_{i,c}$ are the remaining keys in cell $c$ after $i$ steps of peeling using $\tilde g$ and $\sigma_k$ is the corresponding sign of key $k$.
Note that the events $E_{i, c}$ do not depend on whether $g_a$ correctly identified other cells in previous steps as peelable since we consider the peeling process according to $\tilde g$, \emph{not} according to $g_a$.
By \autoref{lem:helpful} we know that 
\begin{align*}
\Pr[E_{i, c}] = \Pr\mleft[g_a\Bigl(\sum_{k \in \vec{k}_{i,c}}\sigma_k k\Bigr) = \sum_{k \in \vec{k}_{i,c}}\sigma_k g_a\left(k\right) \bmod q\mright]\leq& \frac{\delta(2\abs{\vec{k}_{i,c}} p + 1)}{2 Cn^3 \lg(1/\delta)\lglg(1/\delta) p}\\
\leq &\frac{\delta n p}{C \cdot n^3 \lg(1/\delta)\lglg(1/\delta) p}\\
\leq &\frac{\delta}{C n^2\lg(1/\delta)\lglg(1/\delta)},
\end{align*}
where the randomness is taken over the choice of $a$.
By union bounding over all $n$ peeling steps and all $Cn \lg(1/\delta)\lglg(1/\delta)$ cells of the data structure we obtain an \emph{additional} error of at most $\delta$.
Adding this error to the error derived from \autoref{thm:mainiblt} yields the theorem statement.
\qed
\end{proof}

\bibliographystyle{alphaurl}
\bibliography{bib/abbrev0,bib/crypto,bib/extrarefs}

\appendix

\section{Encrypted Compression}\label{app:encryptedcompression}

In this section, we show that the approach towards compressing encrypted data of Fleischhacker, Larsen, and Simkin~\cite{EC:FleLarSim23} is compatible with our new stacked IBLT.
In addition, we generalize their construction to work with encryption schemes that have arbitrarily small plaintext spaces.

\subsection{Additional Preliminaries}

For a set $X^n$, we use the scissor operator $\unique(X^n) := \{ (x_1, \dots, x_n) \in X^n \mid x_i \neq x_j \ \forall i,j \in [n]\}$ to denote the subset of $X^n$ consisting only of those vectors with unique entries.

\begin{definition}[Sparse Vector Representation]
  Let $\FF_q$ be a field and let $\vec{a}\in\FF_q^n$ be a vector.
  The sparse representation of $\vec{a}$ is the set $\sparse(\vec{a}) := \{(i,a_i)\mid a_i\neq 0\}$.
\end{definition}

\subsubsection{Homomorphic Encryption}

Informally, a homomorphic encryption scheme allows to compute an encryption of $f(\vec{m})$ given only the description of $f$ and an encryption of $\vec{m}$.
Throughout the paper, we assume that functions are represented as circuits composed of addition and multiplication gates.
We recall the formal definition of a homomorphic encryption scheme, closely following the notation of~\cite{EC:FleLarSim23}.

\begin{definition}
A homomorphic encryption scheme $\encscheme$ is defined by a tuple of PPT algorithms $(\gen, \enc, \eval,\allowbreak \dec)$ that work as follows:
\begin{description}
	\item[$\gen(\secparam)$:] The key generation algorithm takes the security parameter $\secparam$ as input and returns a secret key $\sk$ and public key $\pk$. The public key implicitly defines a message space $\msgspace$ and ciphertext space $\cspace$. We denote the set of all public keys as $\pkspace$.
	\item[$\enc(\pk, m)$:] The encryption algorithm takes the public key $\pk$ and message $m \in \msgspace$ as input and returns a ciphertext $c \in \cspace$.
	\item[$\eval(\pk, f, \vec{c})$:] The evaluation algorithm takes the public key $\pk$, a function $f : \msgspace^n \to \msgspace^m$, and a vector $\vec{c} \in \cspace^n$ of ciphertexts as input and returns a new vector of ciphertexts $\tilde{\vec{c}} \in \cspace^m$.
	\item[$\dec(\sk, c)$:] The deterministic decryption algorithm takes the secret key $\sk$ and ciphertext $c \in \cspace$ as input and returns a message $m \in \msgspace \cup \{\bot \}$.
\end{description}
\end{definition}

Throughout the paper it is assumed that the ciphertext size is fixed and does not increase when applying the homomorphic evaluation algorithm.
We extend the definition of $\enc$ and $\dec$ to vectors and matrices of messages and ciphertexts respectively, by applying them componentwise,
i.e., for any matrix $\vec{M}\in\msgspace^{n\times m}$, we have $\enc(\pk,\vec{M}) = \vec{C}$ with $\vec{C} \in \cspace^{n\times m}$ and $C[i,j] = \enc(\pk,M[i,j])$ and equivalently $\dec(\sk,\vec{C}) = \vec{M}'$ with $\vec{M}' \in \msgspace^{n\times m}$ and $M'[i,j] = \dec(\sk,C[i,j])$.
This also applies recursively when, for instance, decrypting a vector of matrices of ciphertexts.
Let $\encscheme$ be an additively homomorphic encryption scheme with message space $\msgspace=\FF_q$ for some prime power $q$.
Let $f : \FF_q^2 \to \FF_q$, $f(a,b) := a+b$ and let $g_\alpha : \FF_q \to \FF_q$, $g(a) := \alpha\cdot a$ for any constant $\alpha\in\FF_q$.
For notational convenience we write $\eval(\pk,f,(c_1,c_2)^\intercal)$ as $c_1 \boxplus c_2$ and $\eval(\pk,g_\alpha,c)$ as $\alpha \boxdot c$ with $\pk$ being inferrable from context.
We naturally extend these notions to vectors, i.e. for two vectors $\vec{c},\vec{c}' \in \cspace^n$ we denote $\vec{c}\boxplus\vec{c}' = (c_0 \boxplus c'_0,\dots,c_n \boxplus c'_n)^\intercal$ and for a vector $\vec{\alpha} \in \msgspace^n$ we denote $\vec{\alpha}\boxdot c = (\alpha_0\boxdot c,\dots,\alpha_0\boxdot c)^\intercal$.
For the sake of simplicity we restrict ourselves to homomorphic encryption schemes with unique secret keys, i.e. for a given $\pk$, there exists at most one $\sk$, such that $(\sk,\pk)\gets\gen(\secparam)$.
The unique secret key is denoted as $\gen^{-1}(\pk)$ and we stress that the function $\gen^{-1}(\cdot)$ does not need to be efficiently computable.

We recall the definition of ciphertexts valid relative to a class of circuits and of a ciphertext compression scheme from \cite{EC:FleLarSim23}.

\begin{definition}[$\circspace$-Validity]\label{def:circval}
  Let $(\gen, \enc, \eval, \dec)$ be a homomorphic encryption scheme, let $\circspace$ be a class of circuits, and let $\pk$ be a public key.
  A vector $\vec{c}$ of ciphertexts is $\circspace$-valid for $\pk$, iff for all functions $f\in\circspace$ it holds that $\bot\notin\dec(\gen^{-1}(\pk),\vec{c})$ and $\dec(\gen^{-1}(\pk),\eval(\pk,f,\vec{c}) = f(\dec(\sk,\vec{c}))$.
  We denote by $\valid(\circspace,\pk)$ the set of ciphertext vectors $\circspace$-valid for $\pk$.
\end{definition}

\begin{definition}[Ciphertext Compression Scheme]\label{def:compenc}
Let $\encscheme=(\gen,\allowbreak \enc,\allowbreak \eval, \allowbreak \dec)$ be a homomorphic public key encryption scheme with ciphertext size $\xi=\xi(\secpar)$.
Let $\pkspace$ be the public key space of $\encscheme$.
For each $\pk\in\pkspace$ let $\filtered_\pk$ be a set of ciphertext vectors.
A $\delta$-compressing, $(1-\epsilon)$-correct ciphertext compression scheme for the family $\filtered:=\{\filtered_\pk \mid \pk\in\pkspace\}$ is a pair of PPT algorithms $(\compress, \decompress)$, such that for any $(\sk,\pk)\gets\gen(\secparam)$ and any $\vec{c}\in\filtered_\pk$ the output length of $\compress(\pk, \vec{c})$ is at most $\delta\xi\abs{\vec{c}}$ and it holds that
\[
  \Pr[\decompress(\sk, \compress(\pk, \vec{c})) = \sparse(\dec(\sk, \vec{c}))] = 1-\epsilon(\secpar),
\]
where the probability is taken over the random coins of the compression and decompression algorithms.
\end{definition}

Just like the construction of \cite{EC:FleLarSim23}, our construction described in \autoref{sec:iblt-construction} works for ciphertext vectors of low Hamming weight, which allow for the homomorphic evaluation of inner product functions.
The following two definitions, taken verbatim from~\cite{EC:FleLarSim23} are recalled in the following.

\begin{definition}[Inner Product Functions]\label{def:ipf}
  The class of inner product functions is the set of functions $\ipcirc = \{f_{\vec{a}} \mid \vec{a} \in \FF_q^n\}$ with
  \[
    f_{\vec{a}} : \FF_q^n \to \FF_q,\quad f_{\vec{a}}(\vec{x}):=\langle\vec{a},\vec{x}\rangle.
  \]
\end{definition}

\begin{definition}[$\ipcirc$-Valid Low Hamming Weight Ciphertext Vectors]
  Let $\encscheme = (\gen, \enc, \eval, \dec)$ be a homomorphic public key encryption scheme.
  For any $\pk\in\pkspace$, let
  \[
    \filtered^{\mathsf{ip}}_{t,\pk} := \bigl\{
      \vec{c} \in \valid(\ipcirc,\pk) \mid \hw(\dec(\gen^{-1}(\pk),\vec{c})) < t
    \bigr\}.
  \]
  We then define the family of $\ipcirc$-valid ciphertext vectors with low hamming weight as $\filtered^{\mathsf{ip}}_t := \{\filtered^{\mathsf{ip}}_{t,\pk}\mid\pk\in\pkspace\}$.
\end{definition}

\subsection{Pseudorandom Functions with Variable Codomains}
The construction presented in \autoref{sec:iblt-construction} relies on a pseudorandom function that needs to be able to produce outputs from variable codomains.
We define such a variant of PRFs here.
\begin{definition}[Pseudorandom Function with Variable Codomain]\label{def:prf}
  An efficiently computable function $\prf : \bin^\secpar \times \NN \times \bin^* \to \NN$ is a pseudorandom function with variable codomain, if it satisfies the following properties.
  \begin{enumerate}
    \item For any $s\in\bin^\secpar$, any $\gamma\in\NN$ with $\log \gamma = \poly$, and any $x\in\bin^*$, it holds that $\prf(s,\gamma,x) \in [\gamma]$.
    \item Let $\mathcal{G}$ be the set of all functions $g : \NN \times \bin^* \to \NN$ such that for all $\gamma\in\NN$ and all $x\in\bin^*$ it holds that $g(\gamma,x)\in[\gamma]$.
    For all PPT adversaries $\adv$ it holds that
      \[
        |\Pr[\adv^{\prf(s,\cdot,\cdot)}(\secparam) = 1] - \Pr[\adv^{g(\cdot,\cdot)}(\secparam) = 1]| \leq \negl
      \]
      where the probabilities are taken over the uniform choice of $s\in\bin^\secpar$ and $g\in\mathcal{G}$ respectively.
  \end{enumerate}
\end{definition}
  While this funky definition of a PRF is helpful to us as an abstraction, such PRFs are luckily existentially equivalent to regular PRFs.
  To see this, consider a regular PRF $\prf' : \bin^\secpar \times \bin^* \to \bin^\secpar$.
  We can construct a PRF with variable codomain $\prf$ as follows.
  On input $(s,\gamma,x)$ first compute $s' := \prf'(s,\gamma)$.
  This step gives us (computationally) independent keys for the PRF evaluations for different output domains.
  Then compute $y' := \prf'(s',x)$, this already gives us a pseudorandom value however it's from the wrong domain.
  We can now stretch $y'$ to a sufficient length using a pseudorandom generator and finally reduce it modulo $\gamma$ to get a pseudorandom value in $[\gamma]$.
  A simple hybrid argument can be used to establish pseudorandomness.

\subsection{Wunderbar Pseudorandom Vectors over $K\in\FF_q^\eta$}\label{sec:wunderbar}

As in the original construction, the ciphertext compression scheme relies on wunderbar pseudorandom vectors.
The construction requires that the the universe $K$ over which the wunderbar pseudorandom vector operates is \enquote{large enough}.
In \cite{EC:FleLarSim23} this was achieved by requiring that the field the encryption scheme operates on is large.
Here we show how the same can be achieved by instead defining $K\subseteq \FF_q^\eta$ for an arbitrarily small $q$ and large enough $\eta$.
We first recall the definition of a a wunderbar pseudorandom vector taken verbatim from \cite{EC:FleLarSim23}.

\begin{definition}
  A pseudorandom vector with index recovery for an efficiently sampleable universe $K=K(\secpar)$ consists of a triple of ppt algorithms $(\sample,\allowbreak\entry,\allowbreak\recindex)$ such that
  \begin{description}
    \item[$\sample(\secparam,1^n)$:] The sampling algorithm takes as input the security parameter $\secpar$ and the vector length $n$ in unary and outputs the description of a pseudorandom vector $s$.
    \item[$\entry(s,i)$:] The deterministic retrieving algorithm takes as input a description $s$ and an index $i\in[n]$ and outputs a value $k_i \in K$.
    \item[$\recindex(s,k)$:] The deterministic index recovery algorithm takes as input a description $s$ and a value $k$ and outputs either an index $i\in[n]$ or $\bot$.
  \end{description}
  A pseudorandom vector with index recovery is correct, if for all vector lengths $n=\poly$ and all seeds $s\gets\sample(\secparam,1^n)$ it holds that:
  \begin{enumerate}
  \item For all indices $i\in[n]$ it holds that $\recindex(s,\entry(s,i))=i$.
  \item For all all $k^*\not\in \{\entry(s,i) \mid i\in[n]\}$ it holds that $\recindex(s,k^*)=\bot$.
  \end{enumerate}
  The pseudorandom vector is \emph{wunderbar} if the description of a vector has length $\bigO{\secpar}$ and the runtime of $\entry$ and $\recindex$ is $\bigO{\polylog(n)}$.
  A pseudorandom vector is secure, if for all $n=\poly$ and all ppt algorithms $\adv$
  \[
    \abs{\Pr\mleft[\begin{aligned}s\gets\sample(\secparam,1^n),\\ \vec{k} := \begin{pmatrix}\entry(s,1)\\\vdots\\\entry(s,n)\end{pmatrix}\end{aligned} : \adv(\vec{k})\mright] - \Pr[\vec{k} \gets \unique(K^n) : \adv(\vec{k})]}\leq \negl
  \]
\end{definition}

Fleischhacker, Larsen, and Simkin~\cite{EC:FleLarSim23} construct a wunderbar pseudorandom vector for $K\subseteq \FF_q$ from a pseudorandom permutation.
The construction essentially just takes a PRP over $\FF_2^\secpar$ and uses an efficiently computable and invertible injective function to map values from $\FF_2^\secpar$ to $\FF_q$ and back.
The construction is easily generalized for $K\subseteq S$ for any set $S$ as long as there exists an efficiently computable and invertible injective function from $\FF_2^\secpar$ to $S$.

The new construction requires $K \subseteq \FF_q^\eta$ for some $\eta$ and such that $|K| > \alpha$ for some given lower bound $\alpha$. We specify the required injective function in the following.

Let $\decomp_p : \NN \to [p]^*$ denote the function that maps an integer to its canonical $p$-ary representation and let $\proj_p : [p]^* \to \NN$ be its inverse.
Let $q=p^m$ be an arbitrary prime power and let $\eta = \lceil \secpar/\log q \rceil= \lceil \secpar/(m\log p) \rceil$
We then define an injective function 
\[
  \bintofield : \FF_2^\secpar \to \FF^\eta_{q} \quad \bintofield(\vec{b}) = \vec{d}
\]
where
\[
  d_i := \sum_{j=0}^{m-1} c_{im+j}x^j
\]
where
\[
  \vec{c} := \decomp_q(\proj_2(\vec{b})).
\]
We further specify the inverse function as
\begin{gather*}
  \fieldtobin : \FF^\eta_{q} \to \FF_2^\secpar\cup\{\bot\}\\
\fieldtobin(\vec{d}) := \begin{cases}\bot & \text{if } \proj_p(\vec{c}) \geq 2^\secpar\\\decomp_2(\proj_q(\vec{c}))&\text{otherwise}\end{cases}
\end{gather*}
where
\[
  d_i = \sum_{j=0}^{m-1} c_{im+j}x^j.
\]
For a given $\alpha$, any $\secpar = \Omega(\log \alpha)$ leads to the required wunderbar pseudorandom vector.


\subsection{A Ciphertext Compression Scheme for Small Fields}\label{sec:iblt-construction}
In this section we present a construction of a ciphertext compression scheme, that in contrast to \cite{EC:FleLarSim23} also works if the encryption scheme is defined over an arbitrarily small field, even for $\FF_2$.

\subsection{The Generalized Helpful Lemma}\label{sec:helpful}
Fleischhacker, Larsen, and Simkin~\cite{EC:FleLarSim23} state the following helpful lemma.
\begin{lemma}[Helpful Lemma~{\cite[Lemma~13]{EC:FleLarSim23}}]\label{lem:helpful}
  Let $K \subseteq \FF_q$, $(m_1,\dots,\allowbreak m_n)\in\FF_q^n$ and $I\subseteq [n]$ be arbitrary such that $\sum_{i\in I}m_i \neq 0$ and there exist $i,i'\in I$ with $0\not\in\{m_i,m_{i'}\}$.
  It holds that
  \[
    \Pr\biggl[\vec{k} \gets K^n : \exists j\in[n]\ldotp k_j=\frac{\sum_{i\in I} k_im_i}{\sum_{i\in I} m_i}\biggr]\leq \frac{n}{\abs{K}}
  \]
\end{lemma}
We generalize this lemma to vectors of $\FF_q$ elements.
\begin{lemma}[Generalized Helpful Lemma]\label{lem:genhelpful}
  Let $\eta\in\NN^+$, $K \subseteq \FF^\eta_q$, $(m_1,\dots,\allowbreak m_n)\in\FF_q^n$, and $I\subseteq [n]$ be arbitrary such that $\sum_{i\in I}m_i \neq 0$ and there exist $i,i'\in I$ with $0\not\in\{m_i,m_{i'}\}$.
  It holds that
  \[
    \Pr\biggl[\vec{k} \gets K^n : \exists j\in[n]\ldotp \vec{k}_j = \bigl(\sum_{i\in I} m_i\bigr)^{-1}\cdot\sum_{i\in I} m_i\cdot\vec{k}_i\biggr]\leq \frac{n}{\abs{K}}
  \]
\end{lemma}
\begin{proof}
  Observe that $\vec{k}_j \in K\subseteq \FF_q^\eta$ can be interpreted as polynomials of degree at most $\eta-1$ with coefficients in $\FF_q$.
  Similarly, $m_i\in\FF_q$ is simply a constant polynomial over $\FF_q$ and the vector-scalar multiplications are in fact correct polynomial multiplications resulting in polynomials of degree at most $\eta-1$ with coefficients in $\FF_q$.
  Therefore, the lemma can be reinterpreted as working over the extension field $\FF_{q^\eta}$.
  It then follows directly as a special case of \autoref{lem:helpful} for $\FF_{q^\eta}$.\qed
\end{proof}
As in \cite{EC:FleLarSim23}, the following corollary follows from the observation that due to the birthday bound the statistical distance between sampling from $K^n$ and $\unique(K^n))$ is at most $n^2/\abs{K}$.
\begin{corollary}\label{cor:sumsdonotequal}
  Let $\eta\in\NN^+$, $K \subseteq \FF^\eta_q$, $(m_1,\dots,m_n)\in\FF_q^n$, and $I\subseteq [n]$ be arbitrary such that $\sum_{i\in I}m_i \neq 0$ and there exist $i,i'\in I$ with $0\not\in\{m_i,m_{i'}\}$.
  It holds that
  \[
    \Pr\biggl[\vec{k} \gets K^n : \exists j\in[n]\ldotp \vec{k}_j = \frac{\sum_{i\in I} k_im_i}{\sum_{i\in I} m_i}\biggr]\leq \frac{n^2+n}{\abs{K}}
  \]
\end{corollary}

\subsection{Construction}

The construction presented here essentially takes the construction of Fleischhacker, Larsen, and Simkin~\cite{EC:FleLarSim23}, applies the improved IBLT construction from this work, and instantiates the wunderbar pseudorandom vector using the construction for $K\subseteq\FF_q^\eta$ described in \autoref{sec:wunderbar}.
We give a full formal proof of the construction here.

Before we give the actual construction we first specify two variants of the stacked IBLT construction from this work and prove several lemmas about them.
These two variants are specified in \autoref{fig:modstackediblt} and \autoref{fig:modbasiciblt}.
\begin{figure}[t]\centering
\begin{pcvstack}
  \begin{pchstack}[center,boxed]
    \procedure{$\encode_1(S, \gamechange{s_1})$}{
      \pcfor 0 \leq i < \log t - \log \tau\\
      \quad F_i := \basicencode_{\gamechange{1}}(1,\lceil Ct2^{-i}\rceil,S,\gamechange{(i,s_1)})\\
      \pcfor 0 \leq i < \log \tau\\
      \quad i' := \lfloor\log t - \log \tau\rfloor\\
      \quad F_{i'} := \basicencode_{\gamechange{1}}(2^i,\lceil C\tau 2^{-i}\rceil,S,\gamechange{(i',s_1)})\\
      \pcreturn (F_0,\dots,F_{\lceil\log t\rceil-1})
    }
    \pchspace
    \procedure{$\decode_1((F_0,\dots,F_{\lceil\log t\rceil-1}),\gamechange{s_1})$}{
      S' := \emptyset\\
      \pcfor 0 \leq i < \lceil\log t\rceil\\
      \quad F_i := \basicdelete_{\gamechange{1}}(F_i,S',\gamechange{(i,s_1)})\\
      \quad S' := S' \cup \basicdecode_{\gamechange{1}}(F_i,\gamechange{(i,s_1)})\\
      \pcreturn S'
    }
  \end{pchstack}
  \begin{pchstack}[center,boxed]
    \procedure{$\encode_2(S, s_1,\gamechange{s_2})$}{
      \pcfor 0 \leq i < \log t - \log \tau\\
      \quad F_i := \basicencode_{\gamechange{2}}(1,\lceil Ct2^{-i}\rceil,S,(i,s_1),\gamechange{s_2})\\
      \pcfor 0 \leq i < \log \tau\\
      \quad i' := \lfloor\log t - \log \tau\rfloor\\
      \quad F_{i'} := \basicencode_{\gamechange{2}}(2^i,\lceil C\tau 2^{-i}\rceil,S,(i',s_1),\gamechange{s_2})\\
      \pcreturn (F_0,\dots,F_{\lceil\log t\rceil-1})
    }
    \pcvspace
    \procedure{$\decode_2((F_0,\dots,F_{\lceil\log t\rceil-1}),s_1,\gamechange{s_2})$}{
      S' := \emptyset\\
      \pcfor 0 \leq i < \lceil\log n\rceil\\
      \quad F_i := \basicdelete_{\gamechange{2}}(F_i,S',(i,s_1),\gamechange{s_2})\\
      \quad S' := S' \cup \basicdecode_{\gamechange{2}}(F_i,(i,s_1),\gamechange{s_2})\\
      \pcreturn S'
    }
  \end{pchstack}
  \end{pcvstack}
  \caption{Variants of the simplified stacked IBLT from this work.
  These use modified basic IBLTs specified in \autoref{fig:modbasiciblt} respectively as a building blocks.
  As with the original stacked IBLT we have $\tau=C_0\kappa$ for a sufficiently large constant $C_0>0$ and $C=8e$. Changes between successive modifications are marked in gray.} \label{fig:modstackediblt}
\end{figure}
\begin{figure}\centering
  \begin{pchstack}[boxed]
    \pseudocode[colsep=0em]{
      \basicencode_1(\rho,\gamma, S, \gamechange{(r,s_1)})                        \< \basicencode_2(\rho,\gamma, S, (r,s_1),\gamechange{s_2})\\[-\baselineskip]
      \rule[-1ex]{5cm}{.4pt}\hspace{2em}                                       \< \rule[-1ex]{5cm}{.4pt}\\
      \vec{M} := 0^{\rho\times \gamma}                             \< \vec{M} := 0^{\rho\times \gamma}\\
      \vec{K} := 0^{\rho\times \gamma}                             \< \vec{K} := \gamechange{({0}^{\eta})^{\rho\times \gamma}}\\
      \vec{C} := 0^{\rho\times \gamma}                             \< \\
      \pcforeach (d,m) \in S                                       \< \pcforeach (d,m) \in S\\
      \quad \pcforeach i \in [\rho]                                \< \quad \pcforeach i \in [\rho]\\
      \quad \quad j := \gamechange{\prf(s_1,\gamma,(r,i,d))}       \< \quad \quad j := \prf(s_1,\gamma,(r,i,d))\\
      \quad \quad \vec{M}[i,j] := \vec{M}[i,j] + m                 \< \quad \quad \vec{M}[i,j] := \vec{M}[i,j] + m\\
                                                                   \< \quad \quad \gamechange{\vec{k} := \entry(s_2,d)}\\
      \quad \quad \vec{K}[i,j] := \vec{K}[i,j] + d                 \< \quad \quad \vec{K}[i,j] := \vec{K}[i,j] + \gamechange{(m \cdot \vec{k})}\\
      \quad \quad \vec{C}[i,j] := \vec{C}[i,j] + 1                 \< \\
      \pcreturn (\vec{M},\vec{K},\vec{C})                          \< \pcreturn (\vec{M},\vec{K})\\[1em]
      \basicdelete_1((\vec{K},\vec{M}, \vec{C}), \tilde S, \gamechange{(r,s_1)})  \< \basicdelete_2((\vec{K},\vec{M}), \tilde S\gamechange{, (r,s_1,s_2)})\\[-\baselineskip]
      \rule[-1ex]{5cm}{.4pt}                                       \< \rule[-1ex]{5cm}{.4pt}\\
      \pcforeach (d,m) \in \tilde S                                \< \pcforeach (d,m) \in \tilde S\\
      \quad \pcforeach i \in [\rho]                                \< \quad \pcforeach i \in [\rho]\\
      \quad \quad j := \gamechange{\prf(s_1,\gamma,(r,i,d))}       \< \quad \quad j := \gamechange{\prf(s_1,\gamma,(r,i,d))}\\
                                                                   \< \quad \quad \gamechange{\vec{k} := \entry(s_2,d)}\\
      \quad \quad \vec{M}[i,j] := \vec{M}[i,j] - m                 \< \quad \quad \vec{M}[i,j] := \vec{M}[i,j] - m\\
      \quad \quad \vec{K}[i,j] := \vec{K}[i,j] - d                 \< \quad \quad \vec{K}[i,j] := \vec{K}[i,j] - \gamechange{(m \cdot \vec{k})}\\
      \quad \quad \vec{C}[i,j] := \vec{C}[i,j] - 1                 \< \\
      \pcreturn (\vec{M},\vec{K},\vec{M})                          \< \pcreturn (\vec{M},\vec{K})\\[1em]
      \basicdecode_1((\vec{K}, \vec{M}, \vec{C}),\gamechange{(r,s_1)})\< \basicdecode_2((\vec{K}, \vec{M})\gamechange{,(r,s_1,s_2)})\\[-\baselineskip]
      \rule[-1ex]{5cm}{.4pt}                                       \< \rule[-1ex]{5cm}{.4pt}\\
      S' := \emptyset                                              \< S' := \emptyset\\
      \pcfor (i,j) \in [\rho]\times[\gamma]                        \< \pcfor (i,j) \in [\rho]\times[\gamma]\\
                                                                   \< \quad \gamechange{\pcif \vec{M}[i,j] \neq 0}\\
                                                                   \< \quad\quad \gamechange{d := \recindex\bigl(s_2,(\vec{M}[i,j])^{-1}\cdot\vec{K}[i,j]\bigr)}\\
      \quad\pcif \vec{C}[i,j]=1                                    \< \quad\quad \pcif \gamechange{d\in[n]}\\
      \quad \quad S' := S'\cup \{(\vec{K}[i,j],\vec{M}[i,j])\}     \< \quad \quad \quad S' := S'\cup \{(\gamechange{d},\vec{M}[i,j])\}\\
      \pcreturn S'                                                 \< \pcreturn S'
    }
  \end{pchstack}
  \caption{
    The left hand side shows the basic IBLT, but using a PRF with variable codomain as replacement for the truly random functions.
    The difference between the original basic IBLT and this one are marked in gray.
    The right hand side shows a modified basic IBLT that works without a count matrix and allows insertions using only addition and multiplication by constants.
    The differences are again marked in gray.
    As long as all inserted messages are non-zero, the decoding of all three filters will be the same with overwhelming probability.
  }\label{fig:modbasiciblt}
\end{figure}
We now state and prove several lemmas about these two variants.
The first lemma states that the first variant described in \autoref{fig:modstackediblt} still works as expected, when truly random functions are replaced by pseudorandom ones.

\begin{lemma}\label{lem:deconeworks}
  Let $\prf$ be a variable output domain pseudorandom function as defined in \autoref{def:prf}.
  Then for any set $S \subseteq \FF_q\times\FF_q$ with $|S|\leq n$ and such that for all $(i,m),(i',m')\in S$, $i\neq i'$ it holds that
  \[
    \Pr[\decode_1(\encode_1(\rho,\gamma,S,s_1),s_1) = S] \geq 1 - 2^{-\kappa} - \negl
  \]
  where the probability is taken over the uniform choice of $s_1$.
\end{lemma}
\begin{proof}
  The lemma follows from \autoref{thm:mainiblt} and by a simple reduction to the pseudorandom of $\prf$.
  Let $S$ be an arbitrary set.
  We established the claimed bound by constructing an adversary $\adv$ against the pseudorandomness of $\prf$ as follows.
  We then related the success probability of $\adv$, to the probability of $\encode_1$ and $\decode_1$ working as intended.
  On input $\secparam$ and given access to an oracle $o$ that contains either a truly random function of $\prf(s_1,\cdot,\cdot)$, $\adv$ computes
  \[
    S' = \decode(\encode(\rho,\gamma,S,\vec{h}),\vec{h})
  \]
  but replaces invocations of $h_{i,j}(\cdot)$ with queries of the form $o(\gamma_i,(i,j,\cdot))$.
  If $S'=S$, $\adv$ outputs $0$, otherwise it outputs $1$.
  Note that if $o$ contains a truly random function, this perfectly simulates
  \[
    \decode(\encode(\rho,\gamma,S,\vec{h}),\vec{h}).
  \]
  If on the other hand $o$ contains $\prf(s_1,\cdot,\cdot)$, this perfectly simulates
  \[
    \decode_1(\encode_1(\rho,\gamma,S,s_2),s_2).
  \]
  From the pseudorandomness of $\prf$ it follows that
  \[
    \biggl|\begin{aligned}&\Pr[\decode(\encode(\rho,\gamma,S,\vec{h}),\vec{h})=S]\\ -{}& \Pr[\decode_1(\encode_1(\rho,\gamma,S,s_2),s_2)=S]\end{aligned}\biggr| \leq \negl.
  \]
  Combined with \autoref{thm:mainiblt} the lemma immediately follows.\qed
\end{proof}

The second variant of the stacked IBLT construction described in \autoref{fig:modstackediblt} essentially applies the same modification to stacked IBLTs that \cite{EC:FleLarSim23} applied to regular IBLTs.
That is, detecting \enquote{peelable} entries no longer uses a count matrix, but instead uses a wunderbar pseudorandom vector.
The following lemma essentially states that, as long as the encoded set does not contain any zero entries, the two variants of stacked IBLTs will decode the same set with high probability if the wunderbar pseudorandom vector operates over a large enough universe.
\begin{lemma}\label{lem:basicdeconeandtwosame}
  Let $(\entry,\recindex)$ be a wunderbar pseudorandom vector.
  Then, for any $r,\rho,\gamma \in \ZZ$, any PRF key $s_2$, and any set $S \subseteq [n]\times(\FF_q\setminus\{0\})$ such that $|S| \leq t$ and for all distinct $(i,m),(i',m')\in S$, $i\neq i'$ it holds that
  \[  
    \begin{aligned}
    &\Pr\Biggl[
    \begin{aligned}
    &\basicdecode_1(\basicencode_1(\rho,\gamma,S,(r,s_1)),(r,s_1))\\ 
    ={}&\basicdecode_2(\basicencode_2(\rho,\gamma,S,(r,s_1),s_2),(r,s_1),s_2)
    \end{aligned}
    \Biggr]\\ \geq{}& 1 - \frac{\rho\gamma (n^2+n)}{|K|}
    \end{aligned}
  \]
  where the probability is taken over the uniform choice of $s_1$.
\end{lemma}
\begin{proof}
  Let $S_1,S_2$ be the sets decoded by $\basicdecode_1$ and $\basicdecode_2$.
  We consider two types of errors:
  There could be an $(d,m) \in S_1 \setminus S_2$ or an $(d,m) \in S_2 \setminus S_1$.
  
  In the first case, since $\basicdecode_1$ is decoding the element, it must the case that $(d,m)$ is mapped into a cell on its own.
  However, this implies that the corresponding cell in the output of $\basicencode_2$ will contain $m$ in the value matrix and $m\cdot \entry(s_2,d)$ in the key matrix.
  Therefore, since $m\neq 0$ and by the correctness of the wunderbar pseudorandom vector, $\basicdecode_2$ will also decode the same element.
  
  In the second case, it must hold that several entries $m_1,\dots,m_a$ got mapped to the same position, but it so happens that
  \[
    \recindex\bigl(s_2,\bigl(\sum_{i=1}^{a} m_i\bigr)^{-1} \cdot \sum_{i=1}^{a} m_i\cdot\entry(s_2,d)\Bigr) \in [n]
  \]
  by using the pseudorandomness of the wunderbar pseudorandom vector and applying \autoref{cor:sumsdonotequal} we can conclude that this will happen for any particular cell with probability at most $(n^2+n)/|K|$.
  Since there are $\rho\gamma$ cells, the lemma follows by a union bound over the number of cells.\qed
\end{proof}
The following lemma states that deletion works as expected in both variants of the basic IBLTs described in \autoref{fig:modbasiciblt} and used as building blocks in the variants of the stacked IBLT described in \autoref{fig:modstackediblt}.
That is, if a set $S$ is encoded and a subset $\tilde S$ is \emph{deleted} from the encoding, the result is \emph{identical} to a fresh encoding of $S\setminus\tilde S$ in both constructions.
\begin{lemma}\label{lem:deleteworks}
  For any $r,\gamma,\rho\in\ZZ$, any PRF key $s_1$, any wunderbar pseudorandom vector $s_2$, any set $S \subseteq [n]\times\FF_q$ such that for all distinct $(i,m),(i',m')\in S$, $i\neq i'$, and any subset $\tilde S \subseteq S$ it holds that
  \[
    \begin{aligned}
    &\basicdelete_1(\basicencode_1(\rho,\gamma,S,(r,s_1)),\tilde S, (r,s_1))\\ ={}& \basicencode_1(\rho,\gamma,S\setminus\tilde S,(r,s_1))\end{aligned}
  \]
  and
  \[
    \begin{aligned}&\basicdelete_2(\basicencode_2(\rho,\gamma,S,(r,s_1),s_2),\tilde S, r,s_1,s_2)\\ ={}& \basicencode_2(\rho,\gamma,S\setminus\tilde S,(r,s_1),s_2)\end{aligned}
  \]
  
  where the probability is taken over the choice of $\vec{h}$, $s_1$, and $s_2$.
\end{lemma}
\begin{proof}
  The lemma follows easily by observing that deletion exactly subtracts the values that were added during encoding in both cases.\qed
\end{proof}
The following lemma now states that also the second variant of the stacked IBLT described in \autoref{fig:modstackediblt} works as intended, as long as the wunderbar pseudorandom vector operates over a large enough universe $K$ and the set $S$ does not contain any zero entries.
\begin{lemma}\label{lem:decodetwoworks}
  Let $\prf$ be a variable output domain pseudorandom function as defined in \autoref{def:prf}.
  Let $(\entry,\recindex)$ be a wunderbar pseudorandom vector.
  Then there exists a large enough constant $C'> 0$ such that for any set $S \subseteq \FF_q\times\{\FF_q\}$ with $|S|\leq t$ and such that for all distinct $(i,m),(i',m')\in S$, $i\neq i'$ it holds that
  \[
    \begin{aligned}&\Pr[\decode_2(\encode_2(\rho,\gamma,S,s_1,s_2),s_1,s_2) = S]\\ \geq{}& 1 - 2^{-\kappa} - \frac{C'(n^2+n)(t+\kappa\log\kappa)}{|K|} - \negl\end{aligned}
  \]
  where the probability is taken over the uniform choice of $s_1$ and $s_2$.
\end{lemma}
\begin{proof}
  Let $S \subseteq \FF_q\times\{\FF_q\}$ with $|S|\leq t$ be arbitrary.
  Consider the two decoding procedures running in parallel.
  Clearly, for the end result to differ, one of the executions of $\basicdecode_{1/2}$ has to result in different outputs.
  
  Let $0 \leq \tilde \imath < \lceil\log t\rceil$ be an index, such that for all executions of $\basicdecode_{1/2}$ with $i\leq \tilde\imath$ the outputs were identical.
  Let $\tilde S$ be the set $S'$ \emph{before} the $i$th execution of $\basicdecode_{1/2}$.
  Clearly $\tilde S'$ is the same in both cases.
  Since $\basicdecode_1$ decodes elements if and only if they happen to be alone in their cell, $\basicdecode_1$ never causes any false positives and it must always hold that $\tilde S\subseteq S$.
  It thus follows from \autoref{lem:deleteworks}, that the outputs of the $\tilde \imath$th executions of $\basicdecode_{1/2}$ are
  \[
    \basicdecode_1(\basicencode_1(\rho_{\tilde\imath},\gamma_{\tilde\imath},S\setminus\tilde S,(\tilde\imath,s_1)),(\tilde\imath,s_1))
  \]
  and
  \[
    \basicdecode_2(\basicencode_2(\rho_{\tilde\imath},\gamma_{\tilde\imath},S\setminus\tilde S,(\tilde\imath,s_1),s_2),(\tilde\imath,s_1),s_2)
  \]
  for some choice of $\rho_{\tilde\imath}$ and $\gamma_{\tilde\imath}$.
  
  By \autoref{lem:basicdeconeandtwosame} the probability that the output differs is then at most $\rho_{\tilde\imath}\gamma_{\tilde\imath}(n^2+n)/|K|$.
  With a simple union bound over all indices $0\leq \tilde \imath < \log t$ and by observing that the entire datastructure overall has $O(t+\kappa\log\kappa)$ cells
  it then follows that there exists some large enough constant $C'$ such that the output of $\decode_2$ differs from the output of $\decode_1$ with probability at most
  \[
    \frac{n^2+n}{|K|}\cdot\smashoperator{\sum_{0 \leq i < \log t}} \rho_{i}\gamma_{i} \leq \frac{C'(n^2+n)(t+\kappa\log\kappa)}{|K|}.
  \]
  Since by \autoref{lem:deconeworks} the output of $\decode_1$ is correct with probability $1 - 2^{-\kappa} - \negl$, the lemma follows by another union bound.\qed
\end{proof}
We now specify a final variant of the basic encoding procedure in \autoref{fig:encryptedbasicfilter}.
Essentially the only important difference between $\widehat{\basicencode}$ and $\basicencode_2$ is that the former acts on an encrypted version of the encoded set (represented by a vector of ciphertexts).
\begin{figure}[t]\centering
  \begin{pchstack}[boxed]
  \procedure{$\widehat{\basicencode}(\rho,\gamma, \vec{c}, (r,s_1),s_2))$}{
    \vec{M} := \enc{0^{\rho\times \gamma}}\\
    \vec{K} := \enc{({0}^{\delta})^{\rho\times \gamma}}\\
    \pcforeach \gamechange{d \in [|\vec{c}|]}\\
    \quad \pcforeach i \in [\rho]\\
    \quad \quad j := \prf(s_1,\gamma,(r,i,d))\\
    \quad \quad \vec{M}[i,j] := \vec{M}[i,j] \gamechange{\boxplus c_d}\\
    \quad \quad \vec{k} := \entry(s_2,d)\\
    \quad \quad \vec{K}[i,j] := \vec{K}[i,j] \gamechange{\boxplus (c_d \boxdot \vec{k})}
    }
  \end{pchstack}
  \caption{The $\widehat{\basicencode}$ procedure is a modified version of $\basicencode_2$ to allow filling the filter under additively homomorphic encryption. Changes are marked in gray.}
  \label{fig:encryptedbasicfilter}
\end{figure}
This now finally allows us to state the actual ciphertext compression scheme in \autoref{fig:compscheme} and we state the correctness of the compression scheme in \autoref{theo:main}.
\begin{figure}\centering
  \begin{pchstack}[boxed]
    \procedure{$\compress(\pk,\vec{c})$}{
      s_1 \gets \bin^\secpar\\
      s_2 \gets \sample(\secparam,1^n)\\
      \pcfor 0 \leq i < \log t - \log \tau\\
      \quad F_i := \gamechange{\widehat{\basicencode}}(1,\lceil Cn2^{-i}\rceil,\gamechange{\vec{c}},(i,s_1),s_2)\\
      \pcfor 0 \leq i < \log \tau\\
      \quad i' := \lfloor\log t - \log \tau\rfloor\\
      \quad F_{i'} := \gamechange{\widehat{\basicencode}}(2^i,\lceil C\tau 2^{-i}\rceil,\gamechange{\vec{c}},(i',s_1),s_2)\\
      \pcreturn ((F_0,\dots,F_{\lceil\log t\rceil-1}),s_1,s_2)
    }
    \procedure{$\decompress(\sk,(\vec{F},s_1,s_2))$}{
    \vec{F}' := \dec(\sk,\vec{F})\\
      \pcreturn \decode_2(\vec{F'},s_1,s_2))
    }
  \end{pchstack}

  \caption{A ciphertext compression scheme for arbitrary additively homomorphic encryption schemes for $\filtered_{t,\pk}^\mathsf{ip}$.}
  \label{fig:compscheme}
\end{figure}

\begin{theorem}\label{theo:main}
  Let $\encscheme=(\gen,\enc,\dec)$ be an additively homomorphic encryption scheme with plaintext space $\FF_q$ and ciphertext length $\xi=\xi(\secpar)$.
  Let $(\sample,\allowbreak\entry,\recindex)$ be a wunderbar pseudorandom vector with index recovery for a universe $K\subseteq\FF^\eta$ with $|K|\geq C'(n^2+n)(t+\kappa\log\kappa)\cdot 2^\kappa$ for a large enough constant $C'>0$ and let $\prf$ be a pseudorandom function with variable codomain.
  Then $(\compress,\decompress)$ as specified in \autoref{fig:compscheme} is a $(1-2^{-(\kappa-1)}-\negl)$-correct $(\secpar+(t+\kappa\log\kappa)\eta\xi)/(n\xi))$ compressing ciphertext compression scheme for $\filtered^{\mathsf{ip}}_t$.
\end{theorem}
Before we prove this theorem we will state the following simple corollary that follows simply by instantiating the construction with the wunderbar pseudorandom vector from \autoref{sec:wunderbar} and holds for all reasonable encryption schemes with ciphertext size $\xi=\Omega(\secpar)$.
\begin{corollary}
  Let $\encscheme=(\gen,\enc,\dec)$ be an additively homomorphic encryption scheme with plaintext space $\FF_q$ and ciphertext length $\xi=\Omega(\secpar)$.
  Let $\prp$ be a pseudorandom permutation over $\bin^{\kappa+C''+ 2\log n + \log(t+\kappa\log\kappa)}$ for some large enough constant $C''$ and let $\prf$ be a pseudorandom function with variable codomain.
  Then $(\compress,\decompress)$ as specified in \autoref{fig:compscheme} can be instantiated to be a $(1-2^{-(\kappa-1)}-\negl)$-correct and $\tilde O(\frac{\kappa t+\kappa^2}{n\log q})$ compressing\footnote{The soft-O notation $\tilde O( \cdot )$ ignores log factors in $\kappa$ and $n$.} ciphertext compression scheme for $\filtered^{\mathsf{ip}}_t$.
\end{corollary}

\begin{proof}[\autoref{theo:main}]
  First observe that $\compress$ executes exactly $\encode_2$ on the set $S = \{(d,m) \in [n]\times\FF_q \mid \dec(\gen^{-1}(\pk),c_d)\}$ but under homomorphic encryption.
  Each cell in the encoding is computed as the inner product of the ciphertext vector and some plaintext vector.
  It thus follows from the $\ipcirc$ validity of $\vec{c}$, that after decryption step in $\decompress$ we have $\vec{F}' = \encode_2(S,s_1,s_2)$.
  However, since by design any $(d,0)\in S$ does not influence the value of $\vec{F}'$ we have in fact that $\vec{F}' = \encode_2(S',s_1,s_2)$ where $S' := \{(d,m)\in S \mid m \neq 0\} = \sparse(\vec{})$.
  Since $\vec{c} \in \filtered_{t,\pk}^{\mathsf{ip}}$ and thus $|S'| \leq t$, we can apply \autoref{lem:decodetwoworks} that
  \begin{align*}
    &\Pr[\decompress(\sk,\compress(\pk,\vec{c})) = \sparse(\dec(\sk,\vec{c})]\\
    ={}&\Pr[\decode_2(\encode_2(S',s_1,s_2),s_1,s_2) = S']\\
    \geq{}& 1 - 2^{-\kappa} - \frac{O(n^2(t+\kappa\log\kappa))}{|K|} - \negl\\
    ={}&1 - 2^{-\kappa} - \frac{O(n^2(t+\kappa\log\kappa))}{\Omega(n^2(t+\kappa\log\kappa)\cdot 2^\kappa} - \negl\\
    \geq{}& 1 - 2^{-(\kappa-1)} - \negl
  \end{align*}
  as claimed.
  
  To see the compression factor, consider that the output of $\compress$ consists of $s_1$ and $s_2$, both of which have length $O(\secpar)$ as well as the encrypted stacked IBLT without counters.
  The IBLT consists of pairs of value and key matrices.
  The value matrices combined have $O(t+\kappa\log\kappa)$ entries of $1$ ciphertext each and the key matrices combined have $O(t+\kappa\log\kappa)$ entries of $\eta$ ciphertexts each.
  Thus overall the output of $\compress$ has a length of $O(\secpar + (t+\kappa\log\kappa)\eta\xi)$ bits leading to the claimed compression factor.\qed
\end{proof}

\end{document}